\newcommand{\ifLong}[1]{#1}
\newcommand{\ifShort}[1]{}
\newcommand{\ifShortVspace}[1]{}

\ifLong{
  \documentclass[runningheads,orivec]{llncs}
}
\ifShort{
  \documentclass[runningheads,orivec]{llncs}
}

\usepackage{comment}
\ifLong{\includecomment{iflong}\excludecomment{ifshort}}
\ifShort{\includecomment{ifshort}\excludecomment{iflong}}
\usepackage[utf8]{inputenc}
\usepackage[T1]{fontenc}
\usepackage{graphicx}
\usepackage{xifthen}
\usepackage{xfrac}

\usepackage{amsmath,amssymb,mathtools}
\usepackage{tikz}
\usetikzlibrary{petri}
\usetikzlibrary{arrows,patterns}
\usetikzlibrary{backgrounds,scopes}
\usetikzlibrary{arrows, shapes, trees, positioning, decorations.markings, patterns}
\usetikzlibrary{positioning,fit,calc,arrows.meta}
\usetikzlibrary{shapes.geometric}

\pgfdeclarelayer{backback}
\pgfdeclarelayer{background}
\pgfdeclarelayer{foreground}
\pgfsetlayers{background,main,foreground}
\usepackage{etoolbox}
\newcommand{\nnull}{\bar n}
\newcommand{\bpmn}{\textsc{bpmn}\xspace}

\newcommand{\pl}[1][]{%
  P_{#1}%
}
\newcommand{\tr}[1][]{%
  T_{#1}
}
\newcommand{\pre}[1][()]{%
  \vphantom{i}^{\bullet}#1%
}
\newcommand{\post}[1][()]{%
  #1\vphantom{i}^{\bullet}%
}       
\newcommand{\M}[1][]{%
  \ifthenelse{\equal{#1}{'}\OR\equal{#1}{''}\OR\equal{#1}{'''}}
  {m#1}%
  {m_{#1}}%
}
\newcommand{\Net}{%
  \mathcal{N}%
}
\newcommand{\step}[1][]{%
  \ifthenelse{\equal{#1}{'}\OR\equal{#1}{''}}
  {\vec t\boldsymbol{#1}}
  {\vec t}
}
\newcommand{\fire}[1]{%
  \mathrel{[#1\rangle}%
}

\newcommand{\rl}{%
  \rho%
}
\newcommand{\R}{%
  \mathcal{R}%
}

\usepackage{xspace}

\newcommand{\figref}[1]{Fig.~\ref{#1}}

\newcommand{\vect}[3][1]{%
  \langle #2_{#1},\dotsc,#2_{#3}\rangle
}
\newcommand{\vectup}[3][1]{%
  \langle #2^{#1},\dotsc,#2^{#3}\rangle
}
\newcommand{\xect}[3][1]{%
  \langle
  \csname#2\endcsname{#1},\dotsc,\csname#2\endcsname{#3}\rangle
}

\newcommand{\G}[1][]{%
  G_{#1}%
}
\newcommand{\game}{%
  \langle \N, \states, \acts, \probs, \pay
  \rangle
}
\newcommand{\rgame}{%
  \langle \N,\rstates, \racts,\rprobs, \pay
  \rangle
}

\newcommand{\racts}[1][]{%
  \ifthenelse{\equal{#1}{}}%
  {{\mathring A}}
  {{\mathring A}^{#1}}
}

\newcommand{\rstates}{%
  \mathring S%
}
\newcommand{\rprobs}{%
  \mathring q%
}

\newcommand{\strat}[2][{\plr}]{%
  \ifthenelse{\equal{#2}{}\AND\equal{#1}{\plr}}%
  {\pi}%
  {\pi^{#1}%
    \ifthenelse{\equal{#2}{}}%
    {}%
    {(#2,\act[#1])}%
  }%
}
\newcommand{\states}{%
  S%
}
\newcommand{\state}[1][]{%
  \ifthenelse{\equal{#1}{'}}%
  {s'}%
  {s_{#1}{}}%
}

\newcommand{\N}{%
  N%
}
\newcommand{\plr}[1][]{%
  \ifthenelse{\equal{#1}{'}}%
  {j}%
  {i_{#1}{}}%
}
\newcommand{\acts}[1][]{%
  \ifthenelse{\equal{#1}{}}%
  {{A}}
  {{A}^{#1}}
}
\newcommand{\act}[1][]{%
  \xact{#1}
}

\newcommand{\xact}[1]{%
  \ifthenelse{\equal{#1}{}}
  a%
  {a^{#1}}%
}

\newcommand{\fmluptext}[3]{%
  \{#1^{#2}\}_{#2\in#3}%
}

\newcommand{\sigs}[1][]{%
  \ifthenelse{\equal{#1}{}}%
  {\mathcal{M}}%
  {\mathcal{M}_{#1}}%
}
\newcommand{\sig}[2]{%
  M_{#2}^{#1}%
}
\newcommand{\m}[1][]{%
  \ifthenelse{\equal{#1}{}}
  {\vec x}%
  {{\vec x}^{#1}}%
}
\newcommand{\sigd}[1][]{%
   d_{#1}%
}
\newcommand{\device}[2]{%
  \left\{\langle\xxfml{sig}#1#2{\N},\sigd[#2]\rangle\right\}_{#2\in\mathbb{N}}
}
\newcommand{\xfml}[3]{%
  \{\csname#1\endcsname#2\}_{#2\in#3}%
}

\newcommand{\xxfml}[4]{%
  \{\csname#1\endcsname#2#3\}_{#2\in#4}%
}

\newcommand{\xpay}[1]{%
  \ifthenelse{\equal{#1}{}}%
  {u}
  {u^{#1}}
}
\newcommand{\pay}[1][]{%
  \xpay{#1}%
}

\newcommand{\probs}{%
  q
}
\newcommand{\prob}[3]{%
  \probs(#1 \mid {#2,#3})%
}
\newcommand{\rprob}[3]{%
  \rprobs(#1 \mid {#2,#3})%
}

\newcommand{\hist}[1]{%
  \obs{#1}{}%
}
\newcommand{\obs}[2]{%
  H_{#1}^{#2}(\mathcal{D})%
}
\newcommand{\PP}[3][\mathcal{D}]{%
  \mathbf{P}_{%
      #1\ifthenelse{\equal{#1}{}}{}{,}%
      #2,#3}%
}
\newcommand{\epay}[2][\plr]{%
  \bar\gamma^{#1}_{#2}
}

\newcommand{\defeq}{%
  \mathrel{{\mathop:}{=}} %
}

\newcommand{\msg}[2][msg]{%
  \node[rectangle,draw,inner sep=1pt,minimum size=0ex,fill=white,outer sep=0pt] (#1) at (#2) {\phantom{ai}};
  \draw[overlay,shorten <=.5pt,shorten >=.5pt] (#1.north west) -- (#1.center) -- (#1.north east);
}
\newcommand{\minitabular}[2][c]{%
  \begin{tabular}[c]{@{}#1@{}}
    #2
  \end{tabular}
}

\makeatletter
\newcommand{\getxytikz}[3]{%
  \tikz@scan@one@point\pgfutil@firstofone#1\relax
  \edef#2{\the\pgf@x}%
  \edef#3{\the\pgf@y}%
}
\makeatother

\usetikzlibrary{calc}

\def\ncbararm{3ex}
\def\ncbarangle{90}
\tikzset{
  arm/.default=3ex,
  arm/.code={\def\ncbararm{#1}}, 
  angle/.default=0,
  angle/.code={\def\ncbarangle{#1}} 
}

\tikzset{
    ncbar/.style = {to path={
        let
            \p1=($(\tikztotarget)+(\ncbarangle:\ncbararm)$)
        in
            -- ++(\ncbarangle:\ncbararm) coordinate (firstncbarcorner)
            -- ($(\tikztotarget)!(firstncbarcorner)!(\p1)$) %
            coordinate[pos=.5] (ncbarmid)
            coordinate[pos=1] (secondncbarcorner)
            -- (\tikztotarget)\tikztonodes
    }}
}

\newcommand{\X}{%
  X%
}
\newcommand{\Xs}{%
  \mathcal{X}
}

\newcommand{\med}{Mrs.~Medina\xspace}%

\usepackage{marginnote}
\usepackage{soul}
\usepackage{boxedminipage}
\newcommand{\notethis}[2][authorX]{%
  \makebox[0pt][c]{\textcolor{gray}{\smash{\({}^{{}^{\leftrightarrow}}\)}}}%
  \marginnote{%
    \tiny {\ul{#1}}\\
    #2%
  }%
}
\renewcommand{\notethis}[2][authorX]{}

\newcommand{\thnote}[1]{%
  \notethis[TH]{#1}%
}
\newcommand{\iwnote}[1]{%
  \notethis[IW]{#1}%
}
\newcommand{\lpnote}[1]{%
  \notethis[LP]{#1}%
}

\providecommand{\thnote}[1]{%
}
\providecommand{\iwnote}[1]{%
}
\providecommand{\lpnote}[1]{%
}

\usepackage{comment}
\usepackage{wasysym}
\usepackage{soul}
\ifLong{\usepackage{txfonts}} 
\usepackage[final]{microtype}
\usepackage{bpmn-events,bpmn-gateways,bpmn-misc}
\usepackage{fixmath}
\renewcommand{\vec}[1]{\ensuremath{\underbar{\ensuremath{#1}}}}

\PassOptionsToPackage{hyphens}{url}\usepackage[pdftex, colorlinks=true, hyperfootnotes=false, hyperindex=false, bookmarks=false, plainpages=false, pagebackref=false, pdfpagelabels=true, pdfstartview=FitH, linkcolor=blue, citecolor=blue, urlcolor=blue]{hyperref}

\usepackage[textsize=scriptsize,backgroundcolor=yellow!40]{todonotes}

\newcommand\commentout[1]{}

\usepackage{hyphenat}
\newcommand{\executeiffilenewer}[3]{%
\ifnum\pdfstrcmp{\pdffilemoddate{#1}}%
{\pdffilemoddate{#2}}>0%
{\immediate\write18{#3}}\fi%
}

\begin{document}

\ifShort{
  \title{Incentive Alignment of Business Processes}
  \titlerunning{Incentive Alignment of Business Processes}
}
\ifLong{
    \title{Incentive Alignment of Business Processes: \\
      a game theoretic approach
       }
  \titlerunning{Incentive Alignment of Business Processes}
}
\author{%
  Tobias Heindel\orcidID{0000-0003-3371-8564}%
  \and%
  Ingo Weber\orcidID{0000-0002-4833-5921}%
}
%
\institute{%
  Chair of Software and Business Engineering, Technische Universitaet Berlin, Germany %
	\email{\{heindel,ingo.weber
    \}@tu-berlin.de}%
}
\maketitle              
\begin{abstract}
  Many definitions of business processes refer to business goals, value creation, or profits/gains of sorts.  Nevertheless, the focus of formal methods research on business processes,\iwnote{sollen wir statt ``research on business process correctness'' lieber ``formal research on business processes'' schreiben? Der Widerspruch zu den Zielen im vorigen Satz ist sonst nicht so klar.\\
    \ul{TH}\\
    oder ``Nevertheless, formal methods for checking business process correctness, like the well-known soundness property, typically consider the qualitative execution semantics of modeling languages.'' und dann unten ``quantifiable interest''\\
		\ul{IW} my point was to get away from the ``correctness'' and talk about formal BPM research in general.\\
		\ul{IW} made the change; let me know if you're uncomfortable with it.} like the well-known soundness property, lies on correctness with regards to execution semantics of modeling languages.  Among others, soundness requires proper completion of process instances.  However, the question of whether participants have any \reversemarginpar\thnote{quantifiable; IW: not necessary IMO}\emph{interest} in working towards completion (or in participating in the process)\iwnote{yes, no? Maybe not, because we don't answer the question completely?}
		has not been addressed as of yet.

  In this work, we investigate whether inter-organizational business processes give participants incentives for achieving the common business goals -- in short, whether incentives are aligned with the process.  In particular, fair behavior should pay off and efficient completion of tasks should be rewarded.  We propose a game-theoretic approach that relies on algorithms for solving stochastic games from the machine learning community.  We describe a method for checking incentive alignment of process models with utility annotations for tasks, which can be used for \emph{a priori} analysis of inter-organizational business processes.  Last but not least, we show that the soundness property corresponds to a special case of incentive alignment.
\end{abstract}%


\section{Introduction}
\label{sec:intro}

Many definitions of business processes refer to business goals~\cite{Weske-book} or value creation~\cite{Fundamentals-book}, but whether process participants are actually incentivized to contribute to a process has not been addressed as yet. 
	For \emph{intra}-organizational processes, this question is less relevant; motivation to contribute is often based on loyalty, bonuses if the organization performs well, or simply that tasks in a process are part of one's job.
	Instead, economic modeling of intra-organizational processes often focuses on cost, e.g. in activity-based costing~\cite{AMA3rdEdition}, which can be assessed using model checking tools \cite{HerbertSharp12StochMC} or simulation~\cite{Cartelli:2014}. 

For \emph{inter}-organizational business processes, such indirect motivation cannot be assumed.
A prime example of misaligned incentives was the \$2.5B write-off in Cisco's supply chain in April 2001~\cite{NarayananAnanth2004SupplyChainIA}: 
success of the overall supply chain was grossly misaligned with the incentives of individual participants.
(This happened although several game theoretic approaches for analyzing incentive structures are available for the case of supply chains~\cite{cachon2004game}.)
Furthermore, modeling incentives accurately is actually possible in cross-organizational processes, e.g., based on contracts and agreed-upon prices.
%
%
Now, with the advent of \emph{blockchain} technology~\cite{2019-Blockchain-Book}, it is possible to execute cross-organizational business processes or choreographies as smart contracts~\cite{2018-Mendling-TMIS,2016-Weber-BPM}. 
	The blockchain serves as a neutral, participant-independent computational infrastructure, and as such enables collaboration across organizations even in situations characterized by a \emph{lack of trust} between participants~\cite{2016-Weber-BPM}. 
	However, as there is no central role for oversight, it is important that incentives are properly designed in such situations, among other reasons to avoid unintended --possibly devastating-- results, like those encountered by Cisco.
In fact, a main goal of the Ethereum blockchain is, according to its founder Vitalik Buterin, to create ``a better world by aligning incentives''\footnote{\url{https://www.ikiguide.com/ethereum/}, accessed 8-3-2020}.
In this paper, we present a principled framework %
for incentive alignment of inter-organizational business processes. %
We consider \bpmn models with suitable annotation concerning the utility\footnote{%
  We shall use utility functions in the sense of von Neumann and Morgenstern~\cite{morgenstern1953theory}. %
} of activities,  %
very much in the spirit of activity-based costing (\textsc{abc}) \cite[Chapter~5]{AMA3rdEdition}. %
In short, %
fair behavior should pay off %
and participants should be rewarded for %
efficient completion of process instances.
In more detail, %
we shall consider \bpmn models as
stochastic games~\cite{Shapley1095StochGames} %
and formalize incentive alignment as ``good'' equilibria of %
the resulting game. %
Which equilibria are the desirable ones depends on the business goals %
w.r.t.\ which we want align incentives. %
In the present paper, %
we focus on \emph{proper completion} and \emph{liveness} of activities. %
Interestingly, the soundness property~\cite{Aalst97soundness} will be rediscovered as %
the special case of incentive alignment %
within a single organization that rewards completion of every activity. %

  The overall contribution of the paper is %
  a framework for incentive alignment of
  business process models, 
	particularly in inter-organizational settings. 
	Our approach is based on game theory and %
  inspired by advances on the solution of stochastic games 
  from the machine learning community,
  which has developed algorithms for %
  the practical computation of Nash~\cite{PrasadEtAl15NashEqu} and correlated equilibria~\cite{MacDermedIsbell2009NIPS,MacDermedIsbell2011AIII}. %
  The framework focuses on checking incentive alignment as %
  an \emph{a priori} analysis of business processes %
  specified as \bpmn models with activity-based utility annotations. 
  Specifically, we:%
  \begin{enumerate} 
  \item describe a principled method for translating \bpmn-models with activity-based costs to stochastic games \cite{Shapley1095StochGames}
  \item propose a notion of incentive alignment that we prove to be a conservative extension of Van der Aalst's soundness property~\cite{Aalst97soundness}, %
  \item illustrate 
    the approach with %
    \ifLong{an}\ifShort{a simplified} %
    order-to-cash (\textsc{o\oldstylenums{2}c}) process. %
  \end{enumerate} 

  We pick up the idea of incentive alignment for supply chains \cite{cachon2004game} %
  and set out to apply it in the realm of inter-organizational business processes. \iwnote{Oben habe ich das umformuliert: unser Hauptaugenmerk liegt auf inter-org; aber eigentlich geht es ja auch, wenn die verschiedenen Rollen intra-org sind. Da ist es wahrscheinlich weniger sinnvoll, siehe oben; aber technisch betrachtet wäre es möglich, oder? NACHTRAG%
  : wenn wir das anders darstellen wollen, sollten wir irgendwas along the lines of ``wir reden immer von inter-org, schliessen damit aber ggf sinnvolle Fälle von intra-org nicht aus.'' Oder aber konsequent im Papier ändern -- inter-org wird an einigen Stellen verwendet.
	   \ul{TH} in der Tat stimmt es, dass intra-org. mit mehreren internen Rollen auch von incentive alignment profitieren kann, (um das Managment zu entlasten?!)\\
			\ul{IW} Warum auch immer. Wir sagen nicht dass es Sinn macht, nur dass es geht?
			\ul{TH} Klar, so 'ne Bemerkung sollte noch passen und ist sicherlich zweckdienlich 
	} %
From a technical point of view, %
we are interested in extending the model checking tools for cost analysis \cite{HerbertSharp12StochMC} %
  for \bpmn process models to proper collaborations, %
  very much like the model checker \textsc{prism}  has been extended %
  from Markov decision processes to games~\cite{kwiatkowska2016prism}. %
  Last but not least, %
  we put importance on the connection to established concepts from %
  the business process management community: %
  not only do we keep the spirit of the soundness property; %
  we even obtain the rigorous result %
  that incentive alignment is a conservative extension of the soundness property.

  The remainder of the paper is structured as follows. We introduce concepts and notations in \autoref{sec:preliminaries}. On this basis, we formulate two versions of incentive alignment in \autoref{sec:incentive-alignment}. Finally, we draw conclusions in \autoref{sec:concl}.
  \ifLong{The proof of the main theorem can be found in Appendix~\ref{apx:proof}.}
  \ifShort{The proof of the main theorem can be found in the extended version~\cite{BPM2020HeindelWeberExt}.}



\section{Game theoretic concepts and the Petri net tool chest}
\label{sec:preliminaries}
We now introduce %
the prerequisite concepts for stochastic games~\cite{Shapley1095StochGames} and elementary net systems~\cite{RozenbergE96ENS}. %
The main benefit of using a game theoretic approach %
is a short list of candidate  definitions of equilibrium, %
which make precise the idea of a ``good strategy'' %
for rational actors that compete as players of a game. %
We shall require the following two properties of an equilibrium: 
{(1)}~no player can benefit from unilateral deviation from the ``agreed'' strategy %
and %
{(2)}~players have the possibility to base their moves on %
information from a single (trusted) mediator. %
The specific instance that we shall use are
\emph{correlated equilibria}~\cite{Aumann1974CorEq,correlatedEqu87Aumann}
as studied by Solan and Vieille~\cite{Solan2002CorEq}.\footnote{%
  Nash equilibria are a special case, %
  which however have drawbacks %
  that motivate Aumann's work on the more general
  correlated equilibria \cite{Aumann1974CorEq,correlatedEqu87Aumann}. %
}%
\ We take ample space to review the latter two concepts, 
followed by a short summary of the background on Petri nets. %

We use the following basic concepts and notation. %
The cardinality and the powerset of a set~\(M\) are denoted by~\(|M|\) %
and \(\wp M\), respectively. %
The set of real numbers is denoted by~\(\mathbb{R}\) %
and \([0,1]\subseteq \mathbb{R}\) is the unit interval. %
A~probability distribution over a finite or countably infinite set~\(M\) is %
a function \(p \colon M \to [0,1]\) whose %
values are non-negative and sum up to~\(1\), %
in symbols \(\sum_{m\in M}p(m) = 1\). %
The set of all  probability distributions over a set~\(M\) is %
denoted by~\(\Delta(M)\). %

\begin{figure}[thb]
  \centering
  \newlength{\custdist}\setlength{\custdist}{25ex}
  \newlength{\suppdist}\setlength{\suppdist}{27ex}
  \newcommand{\secretlabel}[1]{\smash{\makebox[0pt][l]{\rotatebox{45}{\colorbox{orange}{#1}}}}}
  \renewcommand{\secretlabel}[1]{}
\begin{ifshort}
  \scriptsize
    {\begin{tikzpicture}[minimum size=5ex,>=Triangle,thin,xscale=1.5]%

        \node[overlay,opacity=0,fill=white,draw,%
        minimum size=4.99ex,StartEvent,%
        ] (startCustomer) at (1,0)
        {\secretlabel{startCustomer}}; 
        \node[fill=white,draw,%
        minimum size=4.99ex,StartEvent,%
        label={[name=sendOrderL]above:{\minitabular{
            }}},
        right = 3.3ex of startCustomer%
        ] (orderReq) {\secretlabel{orderReq}};

        \node[fill=white,MessageIntermediateThrowEvent,%
        right = 3.3ex of orderReq%
        ] (waitAvlb) {\secretlabel{waitAvlb}};

        \node[overlay,opacity=0,fill=white,draw,%
        minimum size=4.99ex,MessageStartEvent,%
        label={[name=receiveOrderL,overlay,opacity=0]below:{\minitabular{receive\\order}}},%
        below = \suppdist+\custdist of startCustomer
        ] (startSupplier) {\secretlabel{startSupplier}};

        \node[overlay,opacity=0,rounded corners,draw,rectangle,fill=white,%
        right = 3.3ex of startSupplier,
        rotate=-90,
        anchor=225
        ]
        (checkStock)
        {%
          \begin{tabular}[c]{@{}c@{}}
            check 
            stock
          \end{tabular}
        };
        
        \node[fill=white,draw,%
        minimum size=4.99ex,MessageStartEvent,
        right = 3.3ex of checkStock.135%
        ] (acceptableReq) {\secretlabel{acceptableReq}};
        \node[overlay,opacity=0,fill=none,minimum size=4.99ex,MessageEndEvent,draw,%
        above = 3.3ex of acceptableReq,%
        ]
        (nostock)
        {\secretlabel{nostock}};
        \node[overlay,opacity=0,fill=white,minimum size=4.99ex,MessageIntermediateThrowEvent,fill=white,%
        right = 3.3ex of acceptableReq,%
        ]
        (instock)
        {\secretlabel{instock}};

        \node[overlay,opacity=0,%
        minimum size=4.99ex,MessageIntermediateCatchEvent,fill=white,%
        below = 3.3ex of waitAvlb] (reqRej) %
        {\secretlabel{reqRej}};
        \node[overlay,opacity=0,minimum size=4.99ex,EndEvent,%
        below = 3.3ex of startCustomer,%
        label={[overlay,opacity=0,name=unlucky,anchor=160,overlay,outer sep=0pt,inner sep=0pt]south:{rejected}}%
        ]
        (reqRejEnd)
        {\secretlabel{reqRejEnd}};
        \node[overlay,opacity=0,minimum size=4.99ex,MessageIntermediateCatchEvent,fill=white, %
        right = 3.3ex of waitAvlb
        ]
        (reqAcc) {\secretlabel{reqAcc}};
        \node[fill=white,ParallelGateway,
        right = 0.0001ex of reqAcc]
        (waitAndSee)
        {\secretlabel{waitAndSee}};

        \node[fill=white,%
        minimum size=4.99ex,MessageIntermediateCatchEvent,fill=white,%
        below = 3.3ex of waitAndSee,%
        ] (recInvoice) {\secretlabel{recInvoice}};

        \node[fill=white,ParallelGateway,
        right = 0.0001ex of instock
        ]
        (invoiceAndShip)
        {};
        \node[minimum size=4.99ex,MessageIntermediateThrowEvent,fill=white,
        above = 3.3ex of invoiceAndShip]
        (sendInvoice)
        {\secretlabel{sendInvoice}};
        \node[minimum size=4.99ex,MessageIntermediateThrowEvent,fill=white,
        below = 3.3ex of invoiceAndShip]
        (reqShipper)
        {\secretlabel{reqShipper}};
        \node[fill=white,ParallelGateway,
        right = 4*3.3ex of invoiceAndShip
        ]
        (invoicedAndShipped)
        {\secretlabel{invoicedAndShipped}};
        \path (invoicedAndShipped) -- ++(-1*3.3ex,-4*3.3ex)
        node[fill=white,rounded corners,draw,
        ]
        (doShip){ship\secretlabel{doShip}};


        \path (reqShipper) -- ++([yshift=\suppdist,xshift=4*3.3ex]0,0)
        node[draw,fill=white,minimum size=4.99ex,MessageStartEvent,
        ]
        (startShipper)
        {\secretlabel{startShipper}};



        \node[anchor=west,fill=white,rounded corners,draw,
        right = 3.3ex of startShipper.east,%
        rotate=90
        ]
        (pickUp)
        {pick up\secretlabel{pickUp}};
        \node[anchor=west,fill=white,rounded corners,draw,
        above = 3.3ex of pickUp.east,%
        rotate=90,%
        anchor = west
        ]
        (deliver)
        {deliver\secretlabel{deliver}};

        
        \path (waitAndSee.north) -- ++([yshift=3.3ex,xshift=1.5ex]0,0)
        node[fill=white,rounded corners,draw,%
        anchor=west
        ]
        (receive)
        {receive};
        \path (receive.east) -- ++ ([xshift=1.5ex,yshift=-3.3ex]0,0)
        node[fill=white,ParallelGateway,anchor=north] 
        (waitedAndSeen) {\secretlabel{waitedAndSeen}};
        \node[fill=white,rounded corners,draw,%
        rotate=-90,%
        right = 3.3ex of waitedAndSeen.east,%
        anchor=south
        ]
        (check)
        {check};
        \node[fill=white,ExclusiveGateway,
        right = 3.3ex of check.north
        ]
        (decide)
        {\secretlabel{decide}};
        \path (decide) -- ++ ([xshift=3*3.3ex,yshift=-2.5*3.3ex]0,0)
        node[minimum size=4.99ex,MessageIntermediateThrowEvent,fill=white
        ]
        (refuse)
        {\secretlabel{refuse}};
        \node[fill=white,rounded corners,draw,
        right= 5.3ex of refuse
        ]
        (sendBack)
        {\minitabular{send\\back}\secretlabel{sendBack}};

        \path (decide.north) -- ++([yshift=3.3ex,xshift=2*3.3ex]0,0)
        node[minimum size=4.99ex,MessageIntermediateThrowEvent,fill=white]
        (preHappy)
        {\secretlabel{preHappy}};
        \coordinate 
        [
        right = 2*3.3ex of preHappy%
        ]
        (Happy)
        ;
        {
          \node[minimum size=4.99ex,SignalIntermediateThrowEvent,fill=white,
          ] at (Happy) {};
          \node[minimum size=4.99ex,SignalIntermediateThrowEvent,
          label = below:{\ul{payment}}
          ]
          (happy)
          at (Happy)
          {};
        }
        
        \node[minimum size=4.99ex,EndEvent,
        right = 4*3.3ex of happy,
        label={[name=happyl]below:{success}}
        ]
        (happyEnd)
        {\secretlabel{happyEnd}};

        \path (sendBack) -- (sendBack.east -| happyEnd.south)
        node[minimum size=4.99ex,EndEvent,
        label={[name=nosucc]above:{no success }}]
        (unhappy)
        {};


        \node[fill=white,ExclusiveEventBasedGateway,
        right = 3.3ex of deliver.south  
        ]
        (askCustomer)
        {\secretlabel{askCustomer}};
        \node[minimum size=4.99ex,MessageIntermediateCatchEvent,fill=white,
        below right = 3.3ex of askCustomer,
        ]
        (getReceipt)
        {\secretlabel{getReceipt}};
        \node[minimum size=4.99ex,MessageEndEvent,
        below  = 3.3ex of getReceipt,
        label=left:{ok}
        ]
        (isOK)
        {\secretlabel{isOK}};

        \node[fill=white,minimum size=4.99ex,MessageIntermediateCatchEvent,
        right = 2*3.3ex of askCustomer]
        (damageReported)
        {\secretlabel{damageReported}};


        \node[fill=white,draw,rounded corners,%
        right=5.3ex of damageReported,%
        rotate=-90,%
        anchor=south,
        allow upside down]
        (pickBack)
        {pick up\secretlabel{pickBack}};
        \node[fill=white,rounded corners,draw,
        below = 3.3ex of pickBack.east,%
        rotate=-90,%
        anchor=west,%
        allow upside down]
        (return)
        {\minitabular{return}\secretlabel{return}};
        \node[minimum size=4.99ex,MessageEndEvent,
        right = 2*3.3ex of return.north
        ]
        (reportDamage)
        {\secretlabel{reportDamage}};


        \node[fill=white,ExclusiveEventBasedGateway,
        right = 3.3ex of invoicedAndShipped
        ]
        (awaitCustomerResponse)
        {\secretlabel{awaitCustomerResponse}};


        \node[minimum size=4.99ex,MessageIntermediateCatchEvent,fill=white,%
        below = \suppdist-5*3.3ex of reportDamage]
        (trash)
        {\secretlabel{trash}};
        \node[rounded corners,fill=white,draw,%
        left = 3.5*3.3ex of trash%
        ]
        (takeBack){%
          receive\secretlabel{takeBack}
        };
        \node[overlay,opacity=0,rounded corners,fill=white,%
        below = 3.3ex of takeBack,%
        draw
        ]
        (writeOff)
        {write off\secretlabel{writeOff}};
        \node[minimum size=4.99ex,EndEvent,
        right = 3.3ex of writeOff,
        label={[name=ldamage]right:{damage}}
        ]
        (damageWrittenOff)
        {\secretlabel{damageWrittenOff}};

        \path (awaitCustomerResponse.south) -- ++([xshift=2ex,yshift=-2*3.3ex]0,0)
        node[fill=white,minimum size=4.99ex,MessageIntermediateCatchEvent,
        anchor=west
        ]
        (receipt)
        {\secretlabel{receipt}};
        \coordinate[
        right=4*3.3ex of receipt
        ]
        (waitMoney)
        {}; 
        \node[
        minimum size=4.99ex,SignalIntermediateCatchEvent,fill=white,
        label = below:{\ul{payment}}
        ]
        (waitMoney)
        at (waitMoney)
        {};
        \node[minimum size=4.99ex,EndEvent,
        right = 2*3.3ex of waitMoney,
        label={[name=lsold]right:{sold}}
        ]
        (sold)
        {\secretlabel{sold}};

        \foreach \u/\v in {%
          check/decide,%
          orderReq/waitAvlb,%
          waitAvlb/waitAndSee,%
          waitAndSee/recInvoice,%
          waitedAndSeen/check,%
          happy/happyEnd,%
          refuse/sendBack,%
          sendBack/unhappy,%
          preHappy/happy,%
          receipt/waitMoney,%
          acceptableReq/invoiceAndShip,%
          invoiceAndShip/sendInvoice,%
          invoiceAndShip/reqShipper,%
          takeBack/trash,%
          waitMoney/sold,%
          invoicedAndShipped/awaitCustomerResponse,%
          deliver/askCustomer,%
          getReceipt/isOK,%
          askCustomer/damageReported,%
          damageReported/pickBack,%
          pickBack/return,%
          return/reportDamage,%
          pickUp/deliver%
        }
        \draw[->] (\u) -- (\v);

        \foreach \u/\v in {%
          decide/preHappy,%
          decide/refuse,%
          waitAndSee/receive,%
          awaitCustomerResponse/takeBack,%
          awaitCustomerResponse/receipt,%
          askCustomer/getReceipt,%
          startShipper/pickUp,%
          reqShipper/doShip%
        }
        \draw[->] (\u) |- (\v);

        \foreach \u/\v in {%
          doShip/invoicedAndShipped,%
          receive/waitedAndSeen,%
          recInvoice/waitedAndSeen,%
          sendInvoice/invoicedAndShipped%
        }
        \draw[->] (\u) -| (\v);

        \foreach \u/\v/\c/\p/\d in
        {
          waitAvlb/acceptableReq/orderCoord/.4/--,%
          sendInvoice/recInvoice/invoiceCoord/.5/--,%
          reportDamage/trash/trashCoord/.65/--,%
          sendBack.260/sendBack.260|-pickBack.west/nothing/.5/--%
        }
        \draw[dashed,shorten <=-2pt,{Circle[fill=white]}-{Triangle[fill=white]},rounded corners]
        (\u) \d (\v) coordinate[pos=\p] (\c);


        \draw[->](trash.south) to[ncbar,angle=-90,arm=1ex] (damageWrittenOff.north);
        \draw[dashed,shorten <=-2pt,{Circle[fill=white]}-{Triangle[fill=white]}]
        (refuse) to[ncbar,angle=-90,arm=.32*\custdist,rounded corners] (damageReported) ;
        \coordinate (refusalCoord) at (ncbarmid);
        \draw[dashed,shorten <=-2pt,{Circle[fill=white]}-{Triangle[fill=white]}]
        (return) -| 
        (takeBack.120) ;
        \draw[dashed,shorten <=-2pt,{Circle[fill=white]}-{Triangle[fill=white]}]
        (reqShipper) to [ncbar,arm=8.5ex,angle=0,rounded corners] (startShipper) ;%
        \path (firstncbarcorner) -- (secondncbarcorner) coordinate[pos=.75] (shipThis);
        \draw[dashed,shorten <=-2pt,{Circle[fill=white]}-{Triangle[fill=white]}]
        (deliver.east) to[ncbar,angle=90,arm=3.3ex,rounded corners] (receive);
        \draw[dashed,shorten <=-2pt,{Circle[fill=white]}-{Triangle[fill=white]}]
        (doShip.60) to [ncbar,angle=90,arm=8*3.3ex,rounded corners=2pt] (pickUp.west);
        \draw[dashed,shorten <=-2pt,{Circle[fill=white]}-{Triangle[fill=white]}]
        (isOK.south) to[ncbar,angle=-90,arm=4*3.3ex,rounded corners] (receipt);
        \coordinate (receiptCoord) at (ncbarmid);
        \draw[dashed,shorten <=-2pt,{Circle[fill=white]}-{Triangle[fill=white]}]
        
        (preHappy) to[ncbar,angle=-90,arm=7*3.3ex,rounded corners] (getReceipt);
        \coordinate (ShipReceiptCoord) at (ncbarmid);
        \msg[order]{orderCoord};
        \node[anchor=east] at (order.west){\minitabular{order\\ goods}};
        \msg[invoice]{invoiceCoord};
        \node[anchor=east
        ] (invoiceLabel) at (invoice.west){invoice~};
        \msg[daReceipt]{receiptCoord};
        \node[above=0ex of daReceipt] {{receipt~~~}};
        \msg[custAck]{ShipReceiptCoord};
        \node[left=0ex of custAck] {receipt};
        \msg[shippingReq]{shipThis};
        \node[anchor=east] at (shippingReq.west){\minitabular{postage\\fee}};
        \msg[refusal]{refusalCoord};
        \node[anchor=west] at (refusal.north east){refusal};
        \msg[damagereportletter]{trashCoord};
        \node[left= 0ex of damagereportletter] {damage report};

        
        \begin{pgfonlayer}{background}
          \node [fill=lightgray,rectangle,draw,fit={([yshift=-1ex]recInvoice.south)
            (nosucc)(orderReq)(happyl)(refuse)(sendOrderL)(receive)(happyEnd)([yshift=-3pt]sendBack.south)}] (pool1) {};
          \node[rotate=90,anchor=south,outer sep=0pt] (label1) at (pool1.west){Customer};
          \node [fit={(pool1.south west) (pool1.north west) (label1)},inner sep=0pt,rectangle,draw,thick] {};
        \end{pgfonlayer}
        \begin{pgfonlayer}{background}
          \node [fill=lightgray,rectangle,draw,fit={([xshift=-5.00ex]startShipper.south west)([yshift=-3pt]startShipper.south)(deliver)(return)(reportDamage)(isOK)(pickBack)([xshift=2pt]pickBack.north)}] (pool2) {};
          \node[rotate=90,anchor=south,outer sep=0pt] (label2) at (pool2.west){Shipper};
          \node [fit={(pool2.south west) (pool2.north west) (label2)},inner sep=0pt,rectangle,draw,thick] {};
        \end{pgfonlayer}
        \begin{pgfonlayer}{background}
          \node [fill=lightgray,rectangle,draw,fit={(acceptableReq)(sendInvoice)(sold)(lsold)(trash)(doShip)(ldamage)([yshift=3pt]takeBack.north) }] (pool3) {};
          \node[rotate=90,anchor=south,outer sep=0pt] (label3) at (pool3.west){Supplier};
          \node [fit={(pool3.south west) (pool3.north west) (label3)},inner sep=0pt,rectangle,draw,thick] {};
        \end{pgfonlayer}

        \node[anchor=south west,rotate=90,,outer sep=0pt,inner sep=0pt] (checkOk) at (decide.north) {\tiny OK!};
        \node[anchor=north west,rotate=-90,outer sep=0pt,inner sep=0pt] (checkNotOk) at (decide.south) {\tiny damaged!};
        
        
      \end{tikzpicture}
    }
\end{ifshort}
\begin{iflong}
    \scriptsize
    {\begin{tikzpicture}[minimum size=5ex,>=Triangle,thin]%

        \node[fill=white,draw,%
        minimum size=4.99ex,StartEvent,%
        ] (startCustomer) at (1,0)
        {\secretlabel{startCustomer}}; 
        \node[fill=white,%
        minimum size=4.99ex,MessageIntermediateThrowEvent,%
        label={[name=sendOrderL]above:{\minitabular{send\\order}}},
        right = 3.3ex of startCustomer%
        ] (orderReq) {\secretlabel{orderReq}};

        \node[fill=white,ExclusiveEventBasedGateway,%
        right = 3.3ex of orderReq%
        ] (waitAvlb) {};

        \node[fill=white,draw,%
        minimum size=4.99ex,MessageStartEvent,%
        label={[name=receiveOrderL]below:{\minitabular{receive\\order}}},%
        below = \suppdist+\custdist of startCustomer
        ] (startSupplier) {\secretlabel{startSupplier}};

        \node[rounded corners,draw,rectangle,fill=white,%
        right = 3.3ex of startSupplier,
        rotate=-90,
        anchor=225
        ]
        (checkStock)
        {%
          \begin{tabular}[c]{@{}c@{}}
            check 
            stock
          \end{tabular}
        };
        
        \node[fill=white,ExclusiveGateway,%
        right = 3.3ex of checkStock.135%
        ] (acceptableReq) {};
        \node[fill=none,minimum size=4.99ex,MessageEndEvent,draw,%
        above = 3.3ex of acceptableReq,%
        label=left:{
          \begin{tabular}[c]{@{}c@{}}
            out of\\
            stock
          \end{tabular}
        }
        ]
        (nostock)
        {\secretlabel{nostock}};
        \node[fill=white,minimum size=4.99ex,MessageIntermediateThrowEvent,fill=white,%
        right = 3.3ex of acceptableReq,%
        label=below:{confirm}
        ]
        (instock)
        {};

        \node[%
        minimum size=4.99ex,MessageIntermediateCatchEvent,fill=white,%
        below = 3.3ex of waitAvlb] (reqRej) %
        {\secretlabel{reqRej}};
        \node[minimum size=4.99ex,EndEvent,%
        below = 3.3ex of startCustomer,%
        label={[name=unlucky,anchor=160,overlay,outer sep=0pt,inner sep=0pt]south:{rejected}}%
        ]
        (reqRejEnd)
        {\secretlabel{reqRejEnd}};
        \node[minimum size=4.99ex,MessageIntermediateCatchEvent,fill=white, %
        right = 3.3ex of waitAvlb
        ]
        (reqAcc) {\secretlabel{reqAcc}};
        \node[fill=white,ParallelGateway,
        right = 3.3ex of reqAcc]
        (waitAndSee)
        {\secretlabel{waitAndSee}};

        \node[fill=white,%
        minimum size=4.99ex,MessageIntermediateCatchEvent,fill=white,%
        below = 3.3ex of waitAndSee,%
        ] (recInvoice) {\secretlabel{recInvoice}};

        \node[fill=white,ParallelGateway,
        right = 3.3ex of instock
        ]
        (invoiceAndShip)
        {};
        \node[minimum size=4.99ex,MessageIntermediateThrowEvent,fill=white,
        above = 3.3ex of invoiceAndShip]
        (sendInvoice)
        {\secretlabel{sendInvoice}};
        \node[minimum size=4.99ex,MessageIntermediateThrowEvent,fill=white,
        below = 3.3ex of invoiceAndShip]
        (reqShipper)
        {\secretlabel{reqShipper}};
        \node[fill=white,ParallelGateway,
        right = 3*3.3ex+1ex of invoiceAndShip
        ]
        (invoicedAndShipped)
        {\secretlabel{invoicedAndShipped}};
        \path (invoicedAndShipped) -- ++(-2*3.3ex,-4*3.3ex)
        node[fill=white,rounded corners,draw,
        ]
        (doShip){ship\secretlabel{doShip}};


        \path (reqShipper) -- ++([yshift=\suppdist,xshift=4*3.3ex]0,0)
        node[draw,fill=white,minimum size=4.99ex,MessageStartEvent,
        ]
        (startShipper)
        {\secretlabel{startShipper}};



        \node[anchor=west,fill=white,rounded corners,draw,
        right = 3.3ex of startShipper.east,%
        rotate=90
        ]
        (pickUp)
        {pick up\secretlabel{pickUp}};
        \node[anchor=west,fill=white,rounded corners,draw,
        above = 3.3ex of pickUp.east,%
        rotate=90,%
        anchor = west
        ]
        (deliver)
        {deliver\secretlabel{deliver}};

        
        \path (waitAndSee.north) -- ++([yshift=3.3ex,xshift=1.5ex]0,0)
        node[fill=white,rounded corners,draw,%
        anchor=west
        ]
        (receive)
        {receive};
        \path (receive.east) -- ++ ([xshift=1.5ex,yshift=-3.3ex]0,0)
        node[fill=white,ParallelGateway,anchor=north] 
        (waitedAndSeen) {\secretlabel{waitedAndSeen}};
        \node[fill=white,rounded corners,draw,%
        rotate=-90,%
        right = 3.3ex of waitedAndSeen.east,%
        anchor=south
        ]
        (check)
        {check};
        \node[fill=white,ExclusiveGateway,
        right = 3.3ex of check.north
        ]
        (decide)
        {\secretlabel{decide}};
        \path (decide) -- + ([xshift=3.5*3.3ex,yshift=-2.5*3.3ex]0,0)
        node[minimum size=4.99ex,MessageIntermediateThrowEvent,fill=white
        ]
        (refuse)
        {\secretlabel{refuse}};
        \node[fill=white,rounded corners,draw,
        right= 3.3ex of refuse
        ]
        (sendBack)
        {\minitabular{send\\back}\secretlabel{sendBack}};

        \path (decide.north) -- ++([yshift=3.3ex,xshift=2*3.3ex]0,0)
        node[minimum size=4.99ex,MessageIntermediateThrowEvent,fill=white]
        (preHappy)
        {\secretlabel{preHappy}};
        \coordinate 
        [
        right = 2*3.3ex of preHappy%
        ]
        (Happy)
        ;
        {
          \node[minimum size=4.99ex,SignalIntermediateThrowEvent,fill=white,
          ] at (Happy) {};
          \node[minimum size=4.99ex,SignalIntermediateThrowEvent,
          label = below:{\ul{payment}}
          ]
          (happy)
          at (Happy)
          {};
        }
        
        \node[minimum size=4.99ex,EndEvent,
        right = 3*3.3ex of happy,
        label={[name=happyl]below:{success}}
        ]
        (happyEnd)
        {\secretlabel{happyEnd}};

        \path (sendBack) -- (sendBack.east -| happyEnd.south)
        node[minimum size=4.99ex,EndEvent,
        label=above:{no success}]
        (unhappy)
        {};


        \node[fill=white,ExclusiveEventBasedGateway,
        right = 3.3ex of deliver.south  
        ]
        (askCustomer)
        {\secretlabel{askCustomer}};
        \node[minimum size=4.99ex,MessageIntermediateCatchEvent,fill=white,
        below right = 3.3ex of askCustomer,
        ]
        (getReceipt)
        {\secretlabel{getReceipt}};
        \node[minimum size=4.99ex,MessageEndEvent,
        below  = 3.3ex of getReceipt,
        label=left:{ok}
        ]
        (isOK)
        {\secretlabel{isOK}};

        \node[fill=white,minimum size=4.99ex,MessageIntermediateCatchEvent,
        right = 2*3.3ex of askCustomer]
        (damageReported)
        {\secretlabel{damageReported}};


        \node[fill=white,draw,rounded corners,%
        right=3.3ex of damageReported,%
        rotate=-90,%
        anchor=south,
        allow upside down]
        (pickBack)
        {pick up\secretlabel{pickBack}};
        \node[fill=white,rounded corners,draw,
        below = 3.3ex of pickBack.east,%
        rotate=-90,%
        anchor=west,%
        allow upside down]
        (return)
        {\minitabular{return}\secretlabel{return}};
        \node[minimum size=4.99ex,MessageEndEvent,
        right = 2*3.3ex of return.north
        ]
        (reportDamage)
        {\secretlabel{reportDamage}};


        \node[fill=white,ExclusiveEventBasedGateway,
        right = 3.3ex of invoicedAndShipped
        ]
        (awaitCustomerResponse)
        {\secretlabel{awaitCustomerResponse}};


        \node[minimum size=4.99ex,MessageIntermediateCatchEvent,fill=white,%
        below = \suppdist-5*3.3ex of reportDamage]
        (trash)
        {\secretlabel{trash}};
        \node[rounded corners,fill=white,draw,%
        left = 3.5*3.3ex of trash%
        ]
        (takeBack){%
          receive\secretlabel{takeBack}
        };
        \node[rounded corners,fill=white,%
        below = 3.3ex of takeBack,%
        draw
        ]
        (writeOff)
        {write off\secretlabel{writeOff}};
        \node[minimum size=4.99ex,EndEvent,
        right = 3.3ex of writeOff,
        label={[name=ldamage]right:{damage}}
        ]
        (damageWrittenOff)
        {\secretlabel{damageWrittenOff}};

        \path (awaitCustomerResponse.south) -- ++([xshift=2ex,yshift=-2*3.3ex]0,0)
        node[fill=white,minimum size=4.99ex,MessageIntermediateCatchEvent,
        anchor=west
        ]
        (receipt)
        {\secretlabel{receipt}};
        \coordinate[
        right=2*3.3ex of receipt
        ]
        (waitMoney)
        {}; 
        \node[
        minimum size=4.99ex,SignalIntermediateCatchEvent,fill=white,
        label = below:{\ul{payment}}
        ]
        (waitMoney)
        at (waitMoney)
        {};
        \node[minimum size=4.99ex,EndEvent,
        right = 2*3.3ex of waitMoney,
        label={[name=lsold]right:{sold}}
        ]
        (sold)
        {\secretlabel{sold}};

        \foreach \u/\v in {%
          check/decide,%
          startCustomer/orderReq,%
          waitAvlb/reqRej,%
          orderReq/waitAvlb,%
          waitAvlb/reqAcc,%
          reqAcc/waitAndSee,%
          waitAndSee/recInvoice,%
          reqRej/reqRejEnd,%
          waitedAndSeen/check,%
          happy/happyEnd,%
          refuse/sendBack,%
          sendBack/unhappy,%
          preHappy/happy,%
          receipt/waitMoney,%
          startSupplier/startSupplier-|checkStock.south,%
          checkStock.north|-acceptableReq.west/acceptableReq.west,%
          acceptableReq/nostock,%
          acceptableReq/instock,%
          instock/invoiceAndShip,%
          invoiceAndShip/sendInvoice,%
          invoiceAndShip/reqShipper,%
          takeBack/trash,%
          writeOff/damageWrittenOff,%
          waitMoney/sold,%
          invoicedAndShipped/awaitCustomerResponse,%
          deliver/askCustomer,%
          getReceipt/isOK,%
          askCustomer/damageReported,%
          damageReported/pickBack,%
          pickBack/return,%
          return/reportDamage,%
          pickUp/deliver%
        }
        \draw[->] (\u) -- (\v);

        \foreach \u/\v in {%
          decide/preHappy,%
          decide/refuse,%
          waitAndSee/receive,%
          awaitCustomerResponse/takeBack,%
          awaitCustomerResponse/receipt,%
          askCustomer/getReceipt,%
          startShipper/pickUp,%
          reqShipper/doShip%
        }
        \draw[->] (\u) |- (\v);

        \foreach \u/\v in {%
          doShip/invoicedAndShipped,%
          receive/waitedAndSeen,%
          recInvoice/waitedAndSeen,%
          sendInvoice/invoicedAndShipped%
        }
        \draw[->] (\u) -| (\v);

        \foreach \u/\v/\c/\p/\d in
        {
          nostock/reqRej/orderRej/.5/--,%
          instock/reqAcc/orderAcc/.4/--,%
          reportDamage/trash/trashCoord/.65/--,%
          sendBack.south/sendBack.south|-pickBack.west/nothing/.5/--%
        }
        \draw[dashed,shorten <=-2pt,{Circle[fill=white]}-{Triangle[fill=white]},rounded corners]
        (\u) \d (\v) coordinate[pos=\p] (\c);


        \draw[->](trash.south) to[ncbar,angle=-90,arm=1ex] (writeOff.north);
        \draw[dashed,shorten <=-2pt,{Circle[fill=white]}-{Triangle[fill=white]}]
        (refuse) to[ncbar,angle=-90,arm=.32*\custdist,rounded corners] (damageReported) ;
        \coordinate (refusalCoord) at (ncbarmid);
        \draw[dashed,shorten <=-2pt,{Circle[fill=white]}-{Triangle[fill=white]}]
        (return) -| 
        (takeBack.120) ;
        \draw[dashed,shorten <=-2pt,{Circle[fill=white]}-{Triangle[fill=white]}]
        (reqShipper) to [ncbar,arm=8.00ex,angle=0,rounded corners] (startShipper) ;%
        \path (firstncbarcorner) -- (secondncbarcorner) coordinate[pos=.75] (shipThis);
        \draw[dashed,shorten <=-2pt,{Circle[fill=white]}-{Triangle[fill=white]}]
        (orderReq) to[ncbar,angle=-90,arm=\custdist,rounded corners] (startSupplier);
        \coordinate (orderCoord) at (secondncbarcorner);
        \draw[dashed,shorten <=-2pt,{Circle[fill=white]}-{Triangle[fill=white]}]
        (deliver.east) to[ncbar,angle=90,arm=3.3ex,rounded corners] (receive);
        \draw[dashed,shorten <=-2pt,{Circle[fill=white]}-{Triangle[fill=white]}]
        (doShip.60) to [ncbar,angle=90,arm=8*3.3ex,rounded corners=2pt] (pickUp.west);
        \draw[dashed,shorten <=-2pt,{Circle[fill=white]}-{Triangle[fill=white]}]
        (sendInvoice.west) to[ncbar,angle=180,arm=3.00ex,rounded corners] (recInvoice.west);
        \coordinate (invoiceCoord) at (ncbarmid);
        \draw[dashed,shorten <=-2pt,{Circle[fill=white]}-{Triangle[fill=white]}]
        (isOK.south) to[ncbar,angle=-90,arm=2*3.3ex,rounded corners] (receipt);
        \coordinate (receiptCoord) at (ncbarmid);
        \draw[dashed,shorten <=-2pt,{Circle[fill=white]}-{Triangle[fill=white]}]
        
        (preHappy) to[ncbar,angle=-90,arm=7*3.3ex,rounded corners] (getReceipt);
        \coordinate (ShipReceiptCoord) at (secondncbarcorner);
        \msg[order]{orderCoord};
        \node[anchor=south] at (order.north){\minitabular{order\\ goods}};
        \msg[reject]{orderRej};
        \node[anchor=south,rotate=90] at(reject.west){rejection};
        \msg[accept]{orderAcc};
        \node[rotate=90,anchor=south,outer sep=0pt,inner sep=0pt] at (accept.west){acceptance};
        \msg[invoice]{invoiceCoord};
        \node[outer sep=0pt,inner sep=0pt,rotate=90,anchor=north
        ] (invoiceLabel) at ([xshift=-.5ex]invoice.east){invoice};
        \msg[daReceipt]{receiptCoord};
        \node[left=0ex of daReceipt] {receipt};
        \msg[custAck]{ShipReceiptCoord};
        \node[left=0ex of custAck] {receipt};
        \msg[shippingReq]{shipThis};
        \node[anchor=east] at (shippingReq.west){\minitabular{postage\\fee}};
        \msg[refusal]{refusalCoord};
        \node[anchor=west] at (refusal.east){refusal};
        \msg[damagereportletter]{trashCoord};
        \node[left= 0ex of damagereportletter] {damage report};

        
        \begin{pgfonlayer}{background}
          \node [fill=lightgray,rectangle,draw,fit={([yshift=-1ex]recInvoice.south)
            (unlucky)(happyl)(startCustomer)(refuse)(sendOrderL)(receive)(happyEnd)([yshift=-3pt]sendBack.south)}] (pool1) {};
          \node[rotate=90,anchor=south,outer sep=0pt] (label1) at (pool1.west){Customer};
          \node [fit={(pool1.south west) (pool1.north west) (label1)},inner sep=0pt,rectangle,draw,thick] {};
        \end{pgfonlayer}
        \begin{pgfonlayer}{background}
          \node [fill=lightgray,rectangle,draw,fit={([xshift=-5.00ex]startShipper.south west)([yshift=-3pt]startShipper.south)(deliver)(return)(reportDamage)(isOK)(pickBack)([xshift=2pt]pickBack.north)}] (pool2) {};
          \node[rotate=90,anchor=south,outer sep=0pt] (label2) at (pool2.west){Shipper};
          \node [fit={(pool2.south west) (pool2.north west) (label2)},inner sep=0pt,rectangle,draw,thick] {};
        \end{pgfonlayer}
        \begin{pgfonlayer}{background}
          \node [fill=lightgray,rectangle,draw,fit={(startSupplier)(receiveOrderL)(sendInvoice)(sold)(lsold)(trash)(doShip)(ldamage)([yshift=3pt]takeBack.north) }] (pool3) {};
          \node[rotate=90,anchor=south,outer sep=0pt] (label3) at (pool3.west){Supplier};
          \node [fit={(pool3.south west) (pool3.north west) (label3)},inner sep=0pt,rectangle,draw,thick] {};
        \end{pgfonlayer}

        \node[anchor=south west,rotate=90,,outer sep=0pt,inner sep=0pt] (checkOk) at (decide.north) {\tiny OK!};
        \node[anchor=north west,rotate=-90,outer sep=0pt,inner sep=0pt] (checkNotOk) at (decide.south) {\tiny damaged!};
        
        
      \end{tikzpicture}
    }
\end{iflong}
  \caption{A \ifLong{an}\ifShort{a simplified} order-to-cash process}
  \label{fig:o2c-basic}
\end{figure}
\subsection{Stochastic games, strategies, equilibria}
\label{sec:game-theory}
We proceed by reviewing core concepts and central results for %
stochastic games~\cite{Shapley1095StochGames}, %
introducing notation alongside; %
we shall use examples to illustrate the most important concepts. %
The presentation is intended to be self-contained %
such that no additional references should be necessary. %
However, %
the interested reader might want to consult standard references or
additional material, %
e.g., textbooks~\cite{leyton2008essentials,RubinsteinOsborne1994book}, %
handbook articles~\cite{Jaskiewicz2017handbook}, %
and surveys~\cite{SolanVieille2015PNAS}. %
We start with the central notion. %
\begin{definition}[Stochastic game]
  \label{def:stoch-game}
  A \emph{stochastic game}~\(\G\) is a quintuple
  \(
    \G = \game
  \)
  that consists of
  \begin{itemize}
  \item %
    a finite set of \emph{players} \(\N = \{1,\dotsc,|\N|\}\) %
    (ranged over by \(\plr,\plr['],\plr[n]\), etc.); 
  \item %
    a finite set of \emph{states}~\(\states\) %
    (ranged over by \(\state,\state['],\state[n]\), etc.); %
  \item a finite, non-empty set of \emph{action profiles} \(\acts = \prod_{\plr=1}^{|\N|}  \acts[\plr]\) %
    (ranged over by \(\act,\act_n\), etc.),
    which is the Cartesian product of %
    a player-indexed family~\(\fmluptext{\acts}{\plr}{\N}\) %
    of sets \(\acts[\plr]\), %
    each of which contains the \emph{actions} of the respective player %
    (ranged over by \(\act[\plr],\act[\plr]_n\), etc.); %
  \item %
    a non-empty set of \emph{available actions} \(\acts[\plr](\state)\subseteq \acts[\plr]\), %
    for each state~\(\state \in \states\) and player~\(\plr\); 
  \item %
    probability distributions %
    \(\prob{{\cdot}}{\state}{\act} \in \Delta(\states)\), %
    for each state~\(\state\in\states\) and every action profile~\(\act \in \acts\),
    which map each state~\(\state[']\in\states\) to \(\prob{\state[']}{\state}{\act}\),
    the \emph{transition probability}  
    from state~\(\state\)  to state~\(\state'\) under the action profile~\(\act\);
    and
  \item %
    the payoff vectors \(\pay(\state,\act) = \langle \pay[1](\state,\act), \dotsc, \pay[|\N|](\state,\act)\rangle\), %
    for each state~\(\state\in\states\) and every action profile~\(\act = \xect{xact}{|\N|} \in \acts\).   \end{itemize}
\end{definition}
Note that players always have some action(s) available, 
possibly just a dedicated idle action, 
see e.g.~\cite{kwiatkowska2018automated}. %


\begin{ifshort}
  The \bpmn model of \figref{fig:o2c-basic} can be understood as %
  (a partial specification of) a stochastic game  %
  played by a shipper, a customer, and a supplier. %
  Abstracting from data, precise timings,
  and similar semantic aspects, %
  a state of the game is a state of an instance of the process, %
  which is represented as a token marking of the \bpmn model. 
  The actions of each player are the activities and events in the respective pool, %
  e.g., the \emph{ship} task, %
  which \emph{Supplier} performs after receiving an order from the \emph{Customer} %
  and payment of the postage fee to \emph{Shipper}. 
  \thnote{%
    For sending of messages, %
    we ``cheat'' and push responsibility to the sender. \\%
    IW: how is this ``cheating''? Isn't it an aspect of the translation? If so, it may or may not be worth mentioning. \\
    \ul{TH} well, try to phone someone that does not pick up the phone ! \\
    IW: OK. Well, I'm still happy with what we say otherwise. (Could also be seen as an instance of trying to send a message, but failing to; like when you send a message via TCP and cannot, because the receiving server is offline. (disregarding UDP)
  }
  Action profiles are combinations of actions %
  that can (or must) be executed concurrently. %
  For example, %
  sending the order and receiving the order after the start of the collaboration %
  may be performed synchronously (e.g., via telephone). %
  The available actions of a player in a given state %
  are the tasks or events in the respective pool %
  that can be executed or happen next -- plus the possibility to idle. %
  The transition probabilities for available actions in this \bpmn process are all~\(1\), %
  such that if players choose to execute certain tasks next, %
  they will be able to do so as long as %
  the chosen activities are actually available actions. %
  As a consequence, %
  all other transition probabilities are~\(0\). %

  One important piece of information that is \emph{not} explicitly specified %
  in the \bpmn model is the utility (or payoff) of tasks and events. %
  In general, it is non-trivial to chose utility functions. %
  However, %
  the example is chosen such that there are natural candidates; %
  e.g., postage can be looked up from one's favorite carrier. %

  A single instance of the order-to-cash process exhibits the well-known %
  phenomenon %
  that \emph{Customer} has no incentive to pay. %
  However, %
  we want to stress that -- very much for the same reason! --  
  \emph{Shipper} would not have any good reason to perform delivery, %
  once the postage fee is paid. %
  Thus, %
  besides the single instance scenario, 
  we shall consider an %
  unbounded number of repetitions of the process, %
  but only one active process instance at each point in time.%
  \footnote{%
    We leave the very interesting situation of %
    interleaved execution of several process instances for future work.%
  } %
  Now, %
  the rational reason for the shipper to deliver (and return damaged goods) %
  is expected revenue from future process instances. %
     \begin{figure}[t]
    \centering
    \begin{tikzpicture}[minimum size=5ex,baseline={(s_c.east)}]
      \scriptsize
      \node[fill=white,fill=white,StartEvent,draw] (s_1) at (0,0) {};
      \node[fill=white,ExclusiveEventBasedGateway,
      right = 3ex of s_1] (decide_1) {};
      \node[fill=white,MessageIntermediateCatchEvent,above right = 3ex of decide_1] (snd_1)  {};
      \node[fill=white,MessageIntermediateThrowEvent,right = 3ex of snd_1] (dosnd_1)  {};
      \node[fill=white,MessageIntermediateCatchEvent,below right = 3ex of decide_1] (rcv_1) {};
      \node[fill=white,MessageIntermediateCatchEvent,right = 3ex of rcv_1] (dorcv_1)  {};
      \node[fill=white,draw,rounded corners,right = 2*3ex of dorcv_1,anchor=200] (nope_1) {
        \begin{tabular}[c]{@{}c@{}}
          \smiley\smiley\smiley\\
          surfing
        \end{tabular}
      };

      \node[fill=white,draw,rounded corners,right = 2*3ex of dosnd_1] (suc_1) {
        \begin{tabular}[c]{@{}c@{}}
          work\\ %
          \$\$
        \end{tabular}
      };
      \node[fill=white,ExclusiveGateway,right = 31ex of decide_1] (done_1) {};

      \node[EndEvent,right = 3ex of done_1] (end_1){};
      \foreach \u/\v/\l/\s/\k in {%
        s_1/decide_1///--,%
        decide_1/snd_1///|-,
        decide_1/rcv_1///|-,
        rcv_1/dorcv_1///--,%
        snd_1/dosnd_1///--,%
        dorcv_1/dorcv_1-|nope_1.west///--,%
        dosnd_1/suc_1///--,%
        done_1/end_1///--%
      }
      \draw[-Triangle] (\u) \k node[auto,\s]{\l} (\v);
      \foreach \u/\v/\l/\s in {%
        nope_1.-20/done_1//,%
        suc_1/done_1//%
      }
      \draw[-Triangle] (\u) -| node[auto,\s]{\l} (\v);
      \begin{pgfonlayer}{background}
        \node [fill=lightgray,rectangle,draw,fit={(s_1) (snd_1) (rcv_1) (end_1) (suc_1)}] (pool1) {};
        \node[rotate=90,anchor=south,outer sep=0pt] (label1) at (pool1.west){Bob};
        \node [fit={(pool1.south west) (pool1.north west) (label1)},inner sep=0pt,rectangle,draw,thick] {};
      \end{pgfonlayer}

      \begin{scope}[shift={([yshift=-42ex]0,0)}]
        \node[fill=white,fill=white,StartEvent,draw] (s_2) at (0,0) {};
        \node[fill=white,ExclusiveEventBasedGateway,
        right = 3ex of s_2] (decide_2) {};
        \node[fill=white,MessageIntermediateCatchEvent,below right = 3ex of decide_2] (snd_2)  {};
        \node[fill=white,MessageIntermediateThrowEvent,right = 3ex of snd_2] (dosnd_2)  {};
        \node[fill=white,MessageIntermediateCatchEvent,above right = 3ex of decide_2] (rcv_2) {};
        \node[fill=white,MessageIntermediateCatchEvent,right = 3ex of rcv_2] (dorcv_2)  {};
        \node[fill=white,draw,rounded corners,right = 2*3ex of dorcv_2,anchor=160] (nope_2) {
          \begin{tabular}[c]{@{}c@{}}
            fishing \\
            \smiley\smiley
          \end{tabular}
        };
        \node[fill=white,draw,rounded corners,right = 2*3ex of dosnd_2] (suc_2) {
          \begin{tabular}[c]{@{}c@{}}
            work\\ \$\$\$\smiley
          \end{tabular}
        };
        \node[fill=white,ExclusiveGateway,right = 31ex of decide_2] (done_2) {};

        \node[EndEvent,right = 3ex of done_2] (end_2){};
        \foreach \u/\v/\l/\s/\k in {%
          s_2/decide_2///--,%
          decide_2/snd_2///|-,
          decide_2/rcv_2///|-,
          rcv_2/dorcv_2///--,%
          snd_2/dosnd_2///--,%
          dorcv_2/dorcv_2-|nope_2.west///--,%
          dosnd_2/suc_2///--,%
          done_2/end_2///--%
        }
        \draw[-Triangle] (\u) \k node[auto,\s]{\l} (\v);
        \foreach \u/\v/\l/\s in {%
          nope_2.20/done_2//,%
          suc_2/done_2//%
        }
        \draw[-Triangle] (\u) -| node[auto,\s]{\l} (\v);
        \begin{pgfonlayer}{background}
          \node [fill=lightgray,rectangle,draw,fit={(s_2) (snd_2) (rcv_2) (end_2) (suc_2)}] (pool2) {};
          \node[rotate=90,anchor=south,outer sep=0pt] (label2) at (pool2.west){Alice};
          \node [fit={(pool2.south west) (pool2.north west) (label2) },inner sep=0pt,rectangle,draw,thick] {};
        \end{pgfonlayer}

      \end{scope}

      \scriptsize
      \path (s_1) -- (s_2) coordinate[midway] (midtmp);
      \node[fill=white,StartEvent,draw] (s_c) at ([xshift=-14ex]midtmp) {};
      \node[fill=white,draw,rounded corners,right = 3ex of s_c] (flip_c) {\color{white}{flip}};
      \node[overlay] at (flip_c.center) {\includegraphics[height=5ex]{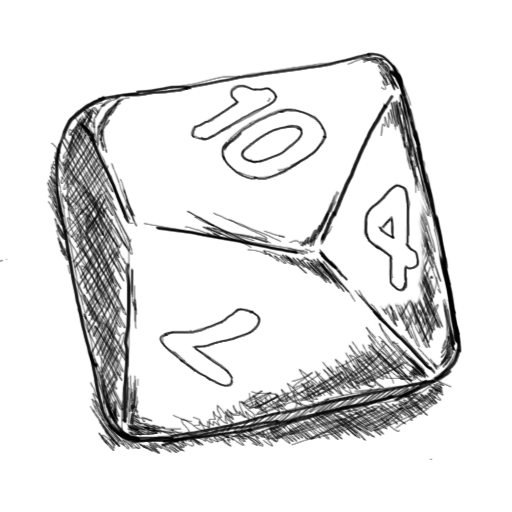}};
      \node[overlay] at ([xshift=-.05pt]flip_c.center) {\includegraphics[height=5ex]{d10.png}};
      \node[overlay] at ([xshift=.05pt]flip_c.center) {\includegraphics[height=5ex]{d10.png}};
      \node[overlay] at ([yshift=-.05pt]flip_c.center) {\includegraphics[height=5ex]{d10.png}};
      \node[overlay] at ([yshift=.05pt]flip_c.center) {\includegraphics[height=5ex]{d10.png}};
      \node[fill=white,ExclusiveGateway,
      right = 3ex of flip_c] (decide_c) {};
      \node[fill=white,MessageIntermediateThrowEvent,above right = 3ex of decide_c] (snd1_c)  {};
      \node[fill=white,MessageIntermediateThrowEvent,right = 2*3ex of snd1_c] (rcv1_c)  {};
      
      \node[fill=white,MessageIntermediateThrowEvent,below right = 3ex of decide_c] (snd2_c)  {};
      \node[fill=white,MessageIntermediateThrowEvent,right = 2*3ex of snd2_c] (rcv2_c)  {};

      \node[fill=white,ExclusiveGateway,below right = 3ex of rcv1_c] (done_c) {};
      \node[EndEvent,right = 2*3ex of done_c] (end_c){};
      \foreach \u/\v/\l/\s/\k in {%
        s_c/flip_c///--,%
        flip_c/decide_c///--,%
        decide_c/snd1_c///|-,%
        decide_c/snd2_c//swap/|-,%
        snd1_c/rcv1_c///--,%
        snd2_c/rcv2_c///--,%
        rcv1_c/done_c///-|,%
        rcv2_c/done_c///-|,%
        done_c/end_c///--%
      }
      \draw[-Triangle] (\u) \k node[auto,\s]{\l} (\v);
      \node[above left = 0ex of decide_c.north]{\({}\leq{6}\)};
      \node[below left = 0ex of decide_c.south] {\({}>{6}\)};

      \begin{pgfonlayer}{background}
        \node [fill=lightgray,rectangle,draw,fit={(s_c) (snd1_c) (rcv2_c) (end_c)}] (poolc) {};
        \node[rotate=90,anchor=south,outer sep=0pt] (labelc) at (poolc.west){\med};
        \node [fit={(poolc.south west) (poolc.north west) (labelc)},inner sep=0pt,rectangle,draw,thick] {};
      \end{pgfonlayer}

      \draw[bend left=30,white,shorten <=-2pt,{Circle[fill=white]}-{Triangle[fill=white]}] (dosnd_1) to (dorcv_2);
      \draw[bend right=35,white,shorten <=-2pt,{Circle[fill=white]}-{Triangle[fill=white]}] (dosnd_2) to (dorcv_1);
      \draw[bend left=30,dashed,shorten <=-2pt,{Circle[fill=white]}-{Triangle[fill=white]}] (dosnd_1) to (dorcv_2);
      \draw[bend right=35,dashed,shorten <=-2pt,{Circle[fill=white]}-{Triangle[fill=white]}] (dosnd_2) to (dorcv_1);

      \foreach \u/\v/\b in {%
        snd1_c.north/snd_1.south/left,%
        rcv1_c.south/rcv_2.north/right,%
        snd2_c.south/snd_2.north/right,%
        rcv2_c.north/rcv_1.south/left%
      }
      \draw[bend \b=20,dashed,shorten <=-2pt,{Circle[fill=white]}-{Triangle[fill=white]}] (\u) to coordinate[pos=.7](msg)  (\v);
    \end{tikzpicture}
    \caption{The \emph{To work or not to work?} collaboration}
    \label{to-work-not-coll-med}
  \end{figure}

  One distinguishing feature of the order-to-cash collaboration is %
  that participants do not need to coordinate with each other in any non-trivial way; %
  in particular, %
  there are no joint decisions to make. %
  To illustrate the point, %
  let us consider a different example. %
  Alice and Bob are co-founders of a company, %
  which is running so smoothly %
  that it suffices when, any day of the week, only one of them is going to work. %


  Alice suggests that their secretary \med could help them out %
  by rolling a 10-sided die 
  each morning and %
  notifying them about who is going to go to work that day,
  dependent on whether the outcome is smaller or larger than six. %
  Alice's idea reasoning behind this elaborate process %
  (as shown in \figref{to-work-not-coll-med}), %
  which lets Bob and Alice work 60\% and 40\% of the days,  %
  resp.,
  is the well-known fact that
  Alice is 50\% more efficient than Bob %
  when it comes to generating revenue. 

\end{ifshort}
\begin{iflong}
  The \bpmn model of \figref{fig:o2c-basic} can be understood as %
  (a partial specification of) a stochastic game  %
  played by a shipper, a customer, and a supplier. %
  Abstracting from data, precise timings,
  and similar semantic aspects, %
  a state of the game is a state of an instance of the process, %
  which is represented as a token marking of the \bpmn model. 
  The actions of each player are the activities and events in the respective pool, %
  e.g., the \emph{check stock} task, %
  which \emph{Supplier} performs upon receiving an order from \emph{Customer}. %
  \thnote{%
    For sending of messages, %
    we ``cheat'' and push responsibility to the sender. \\%
    IW: how is this ``cheating''? Isn't it an aspect of the translation? If so, it may or may not be worth mentioning. \\
    \ul{TH} well, try to phone someone that does not pick up the phone ! \\
    IW: OK. Well, I'm still happy with what we say otherwise. (Could also be seen as an instance of trying to send a message, but failing to; like when you send a message via TCP and cannot, because the receiving server is offline. (disregarding UDP)
  }
  Action profiles are combinations of actions %
  that can (or must) be executed concurrently. %
  For example, %
  sending the order and receiving the order after the start of the collaboration %
  may be performed synchronously (e.g., via telephone). %
  The available actions of a player in a given state %
  are the tasks or events in the respective pool %
  that can be executed or happen next -- plus the possibility to idle. %
  The transition probabilities for available actions in this \bpmn process are all~\(1\), %
  such that if players choose to execute certain tasks next, %
  they will be able to do so as long as %
  the chosen activities are actually available actions. %
  As a consequence, %
  all other transition probabilities are~\(0\). %

  One important piece of information that is \emph{not} explicitly specified %
  in the \bpmn model is the utility (or payoff) of tasks and events. %
  In general, it is non-trivial to chose utility functions. %
  However, %
  the example is chosen such that there are natural candidates; %
  e.g., postage can be looked up from one's favorite carrier. %

  A single instance of the order-to-cash process exhibits the well-known %
  phenomenon %
  that \emph{Customer} has no incentive to pay. %
  However, %
  we want to stress that -- very much for the same reason! --  
  \emph{Shipper} would not have any good reason to perform delivery, %
  once the postage fee is paid. %
  Thus, %
  besides the single instance scenario, 
  we shall consider an %
  unbounded number of repetitions of the process, %
  but only one active process instance at each point in time.%
  \footnote{%
    We leave the very interesting situation of %
    interleaved execution of several process instances for future work.%
  } %
  Now, %
  the rational reason for the shipper to deliver (and return damaged goods) %
  is expected revenue from future process instances. %
  \begin{figure}[t]
    \centering
    \begin{tikzpicture}[minimum size=5ex]
      \scriptsize
      \node[fill=white,fill=white,StartEvent,draw] (s_1) at (0,0) {};
      \node[fill=white,ExclusiveGateway,
      right = 3ex of s_1] (decide_1) {};
      \coordinate[
      above right = 2*3ex of decide_1] (snd_1) ;
      \node[fill=white,MessageIntermediateThrowEvent,right = 0ex of snd_1] (dosnd_1)  {};
      \coordinate[
      below right = 2*3ex of decide_1] (rcv_1);
      \node[fill=white,MessageIntermediateCatchEvent,right = 0ex of rcv_1] (dorcv_1)  {};
      \node[fill=white,draw,rounded corners,right = 3ex of dorcv_1,anchor=200] (nope_1) {
        \begin{tabular}[c]{@{}c@{}}
          \smiley\smiley\smiley\\
          surfing
        \end{tabular}
      };

      \node[fill=white,draw,rounded corners,right = 3ex of dosnd_1] (suc_1) {
        \begin{tabular}[c]{@{}c@{}}
          work\\ %
          \$\$
        \end{tabular}
      };
      \node[fill=white,ExclusiveGateway,right = 20ex of decide_1] (done_1) {};

      \node[EndEvent,right = 3ex of done_1] (end_1){};
      \foreach \u/\v/\l/\s/\k in {%
        s_1/decide_1///--,%
        decide_1/dosnd_1///|-,
        decide_1/dorcv_1///|-,
        dorcv_1/dorcv_1-|nope_1.west///--,%
        dosnd_1/suc_1///--,%
        done_1/end_1///--%
      }
      \draw[-Triangle] (\u) \k node[auto,\s]{\l} (\v);
      \foreach \u/\v/\l/\s in {%
        nope_1.-20/done_1//,%
        suc_1/done_1//%
      }
      \draw[-Triangle] (\u) -| node[auto,\s]{\l} (\v);
      \begin{pgfonlayer}{background}
        \node [fill=lightgray,rectangle,draw,fit={(s_1) (dosnd_1) (dorcv_1) (end_1) (suc_1)}] (pool1) {};
        \node[rotate=90,anchor=south,outer sep=0pt] (label1) at (pool1.west){Bob};
        \node [fit={(pool1.south west) (pool1.north west) (label1)},inner sep=0pt,rectangle,draw,thick] {};
      \end{pgfonlayer}

      \begin{scope}
        [shift={(6,0)}]
        \node[fill=white,fill=white,StartEvent,draw] (s_2) at (0,0) {};
        \node[fill=white,ExclusiveGateway,
        right = 3ex of s_2] (decide_2) {};
        \coordinate[
        below right = 2*3ex of decide_2] (snd_2);
        \node[fill=white,MessageIntermediateThrowEvent,right = 0ex of snd_2] (dosnd_2)  {};
        \coordinate[
        above right = 2*3ex of decide_2] (rcv_2);
        \node[fill=white,MessageIntermediateCatchEvent,right = 0ex of rcv_2] (dorcv_2)  {};
        \node[fill=white,draw,rounded corners,right = 3ex of dorcv_2,anchor=160] (nope_2) {
          \begin{tabular}[c]{@{}c@{}}
            fishing \\
            \smiley\smiley
          \end{tabular}
        };

        \node[fill=white,draw,rounded corners,right = 3ex of dosnd_2] (suc_2) {
          \begin{tabular}[c]{@{}c@{}}
            work\\ \$\$\$\smiley
          \end{tabular}
        };
        \node[fill=white,ExclusiveGateway,right = 20ex of decide_2] (done_2) {};

        \node[EndEvent,right = 3ex of done_2] (end_2){};
        \foreach \u/\v/\l/\s/\k in {%
          s_2/decide_2///--,%
          decide_2/dosnd_2///|-,
          decide_2/dorcv_2///|-,
          dorcv_2/dorcv_2-|nope_2.west///--,%
          dosnd_2/suc_2///--,%
          done_2/end_2///--%
        }
        \draw[-Triangle] (\u) \k node[auto,\s]{\l} (\v);
        \foreach \u/\v/\l/\s in {%
          nope_2.20/done_2//,%
          suc_2/done_2//%
        }
        \draw[-Triangle] (\u) -| node[auto,\s]{\l} (\v);
        \begin{pgfonlayer}{background}
          \node [fill=lightgray,rectangle,draw,fit={(s_2) (dosnd_2) (dorcv_2) (end_2) (suc_2)}] (pool2) {};
          \node[rotate=90,anchor=south,outer sep=0pt] (label2) at (pool2.west){Alice};
          \node [fit={(pool2.south west) (pool2.north west) (label2) },inner sep=0pt,rectangle,draw,thick] {};
        \end{pgfonlayer}

      \end{scope}

      \begin{scope}[overlay]
        \draw[bend left=30,white,shorten <=-2pt,{Circle[fill=white]}-{Triangle[fill=white]}] (dosnd_1) to (dorcv_2);
        \draw[bend left=30,white,shorten <=-2pt,{Circle[fill=white]}-{Triangle[fill=white]}] (dosnd_2) to (dorcv_1);
        \draw[bend left=30,dashed,shorten <=-2pt,{Circle[fill=white]}-{Triangle[fill=white]}] (dosnd_1) to (dorcv_2);
        \draw[bend left=30,dashed,shorten <=-2pt,{Circle[fill=white]}-{Triangle[fill=white]}] (dosnd_2) to (dorcv_1);
      \end{scope}

      \foreach \u/\v/\b in {%
      }
      \draw[bend \b=20,dashed,shorten <=-2pt,{Circle[fill=white]}-{Triangle[fill=white]}] (\u) to coordinate[pos=.7](msg)  (\v);
      \path (end_1) -- (s_2) node[pos=.5,rotate=30,opacity=.4,red]{
        \Huge Bad diagram!
      };
    \end{tikzpicture}
    \caption{Bob's ideas about work life balance with Alice's disapproval}
    \label{fig:work-life-balance}
  \end{figure}
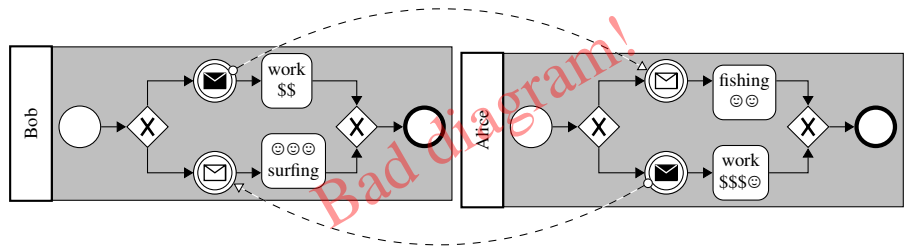

  One distinguishing feature of the order-to-cash collaboration is %
  that participants do not need to coordinate with each other in any non-trivial way; %
  in particular, %
  there are no joint decisions to make. %
  To illustrate the point, %
  let us consider a different example. %
  Alice and Bob are co-founders of a company, %
  which is running so smoothly %
  that it suffices when, any day of the week, only one of them is going to work. %
  Depending on their mood, %
  at least one of them can spend their time with their favorite hobby. %

  \begin{figure}[tbh]
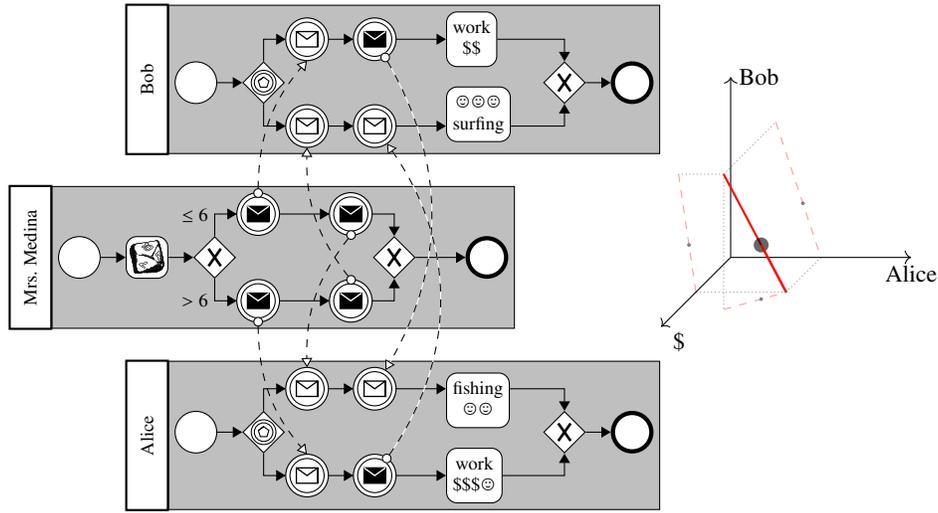

    \centering
    \begin{tikzpicture}[minimum size=5ex,baseline={(s_c.east)}]
      \scriptsize
      \node[fill=white,fill=white,StartEvent,draw] (s_1) at (0,0) {};
      \node[fill=white,ExclusiveEventBasedGateway,
      right = 3ex of s_1] (decide_1) {};
      \node[fill=white,MessageIntermediateCatchEvent,above right = 3ex of decide_1] (snd_1)  {};
      \node[fill=white,MessageIntermediateThrowEvent,right = 3ex of snd_1] (dosnd_1)  {};
      \node[fill=white,MessageIntermediateCatchEvent,below right = 3ex of decide_1] (rcv_1) {};
      \node[fill=white,MessageIntermediateCatchEvent,right = 3ex of rcv_1] (dorcv_1)  {};
      \node[fill=white,draw,rounded corners,right = 2*3ex of dorcv_1,anchor=200] (nope_1) {
        \begin{tabular}[c]{@{}c@{}}
          \smiley\smiley\smiley\\
          surfing
        \end{tabular}
      };

      \node[fill=white,draw,rounded corners,right = 2*3ex of dosnd_1] (suc_1) {
        \begin{tabular}[c]{@{}c@{}}
          work\\ %
          \$\$
        \end{tabular}
      };
      \node[fill=white,ExclusiveGateway,right = 31ex of decide_1] (done_1) {};

      \node[EndEvent,right = 3ex of done_1] (end_1){};
      \foreach \u/\v/\l/\s/\k in {%
        s_1/decide_1///--,%
        decide_1/snd_1///|-,
        decide_1/rcv_1///|-,
        rcv_1/dorcv_1///--,%
        snd_1/dosnd_1///--,%
        dorcv_1/dorcv_1-|nope_1.west///--,%
        dosnd_1/suc_1///--,%
        done_1/end_1///--%
      }
      \draw[-Triangle] (\u) \k node[auto,\s]{\l} (\v);
      \foreach \u/\v/\l/\s in {%
        nope_1.-20/done_1//,%
        suc_1/done_1//%
      }
      \draw[-Triangle] (\u) -| node[auto,\s]{\l} (\v);
      \begin{pgfonlayer}{background}
        \node [fill=lightgray,rectangle,draw,fit={(s_1) (snd_1) (rcv_1) (end_1) (suc_1)}] (pool1) {};
        \node[rotate=90,anchor=south,outer sep=0pt] (label1) at (pool1.west){Bob};
        \node [fit={(pool1.south west) (pool1.north west) (label1)},inner sep=0pt,rectangle,draw,thick] {};
      \end{pgfonlayer}

      \begin{scope}[shift={([yshift=-42ex]0,0)}]
        \node[fill=white,fill=white,StartEvent,draw] (s_2) at (0,0) {};
        \node[fill=white,ExclusiveEventBasedGateway,
        right = 3ex of s_2] (decide_2) {};
        \node[fill=white,MessageIntermediateCatchEvent,below right = 3ex of decide_2] (snd_2)  {};
        \node[fill=white,MessageIntermediateThrowEvent,right = 3ex of snd_2] (dosnd_2)  {};
        \node[fill=white,MessageIntermediateCatchEvent,above right = 3ex of decide_2] (rcv_2) {};
        \node[fill=white,MessageIntermediateCatchEvent,right = 3ex of rcv_2] (dorcv_2)  {};
        \node[fill=white,draw,rounded corners,right = 2*3ex of dorcv_2,anchor=160] (nope_2) {
          \begin{tabular}[c]{@{}c@{}}
            fishing \\
            \smiley\smiley
          \end{tabular}
        };
        \node[fill=white,draw,rounded corners,right = 2*3ex of dosnd_2] (suc_2) {
          \begin{tabular}[c]{@{}c@{}}
            work\\ \$\$\$\smiley
          \end{tabular}
        };
        \node[fill=white,ExclusiveGateway,right = 31ex of decide_2] (done_2) {};

        \node[EndEvent,right = 3ex of done_2] (end_2){};
        \foreach \u/\v/\l/\s/\k in {%
          s_2/decide_2///--,%
          decide_2/snd_2///|-,
          decide_2/rcv_2///|-,
          rcv_2/dorcv_2///--,%
          snd_2/dosnd_2///--,%
          dorcv_2/dorcv_2-|nope_2.west///--,%
          dosnd_2/suc_2///--,%
          done_2/end_2///--%
        }
        \draw[-Triangle] (\u) \k node[auto,\s]{\l} (\v);
        \foreach \u/\v/\l/\s in {%
          nope_2.20/done_2//,%
          suc_2/done_2//%
        }
        \draw[-Triangle] (\u) -| node[auto,\s]{\l} (\v);
        \begin{pgfonlayer}{background}
          \node [fill=lightgray,rectangle,draw,fit={(s_2) (snd_2) (rcv_2) (end_2) (suc_2)}] (pool2) {};
          \node[rotate=90,anchor=south,outer sep=0pt] (label2) at (pool2.west){Alice};
          \node [fit={(pool2.south west) (pool2.north west) (label2) },inner sep=0pt,rectangle,draw,thick] {};
        \end{pgfonlayer}

      \end{scope}

      \scriptsize
      \path (s_1) -- (s_2) coordinate[midway] (midtmp);
      \node[fill=white,StartEvent,draw] (s_c) at ([xshift=-14ex]midtmp) {};
      \node[fill=white,draw,rounded corners,right = 3ex of s_c] (flip_c) {\color{white}{flip}};
      \node[overlay] at (flip_c.center) {\includegraphics[height=5ex]{d10.png}};
      \node[overlay] at ([xshift=-.05pt]flip_c.center) {\includegraphics[height=5ex]{d10.png}};
      \node[overlay] at ([xshift=.05pt]flip_c.center) {\includegraphics[height=5ex]{d10.png}};
      \node[overlay] at ([yshift=-.05pt]flip_c.center) {\includegraphics[height=5ex]{d10.png}};
      \node[overlay] at ([yshift=.05pt]flip_c.center) {\includegraphics[height=5ex]{d10.png}};
      \node[fill=white,ExclusiveGateway,
      right = 3ex of flip_c] (decide_c) {};
      \node[fill=white,MessageIntermediateThrowEvent,above right = 3ex of decide_c] (snd1_c)  {};
      \node[fill=white,MessageIntermediateThrowEvent,right = 2*3ex of snd1_c] (rcv1_c)  {};
      
      \node[fill=white,MessageIntermediateThrowEvent,below right = 3ex of decide_c] (snd2_c)  {};
      \node[fill=white,MessageIntermediateThrowEvent,right = 2*3ex of snd2_c] (rcv2_c)  {};

      \node[fill=white,ExclusiveGateway,below right = 3ex of rcv1_c] (done_c) {};
      \node[EndEvent,right = 2*3ex of done_c] (end_c){};
      \foreach \u/\v/\l/\s/\k in {%
        s_c/flip_c///--,%
        flip_c/decide_c///--,%
        decide_c/snd1_c///|-,%
        decide_c/snd2_c//swap/|-,%
        snd1_c/rcv1_c///--,%
        snd2_c/rcv2_c///--,%
        rcv1_c/done_c///-|,%
        rcv2_c/done_c///-|,%
        done_c/end_c///--%
      }
      \draw[-Triangle] (\u) \k node[auto,\s]{\l} (\v);
      \node[above left = 0ex of decide_c.north]{\({}\leq{6}\)};
      \node[below left = 0ex of decide_c.south] {\({}>{6}\)};

      \begin{pgfonlayer}{background}
        \node [fill=lightgray,rectangle,draw,fit={(s_c) (snd1_c) (rcv2_c) (end_c)}] (poolc) {};
        \node[rotate=90,anchor=south,outer sep=0pt] (labelc) at (poolc.west){\med};
        \node [fit={(poolc.south west) (poolc.north west) (labelc)},inner sep=0pt,rectangle,draw,thick] {};
      \end{pgfonlayer}

      \draw[bend left=30,white,shorten <=-2pt,{Circle[fill=white]}-{Triangle[fill=white]}] (dosnd_1) to (dorcv_2);
      \draw[bend right=35,white,shorten <=-2pt,{Circle[fill=white]}-{Triangle[fill=white]}] (dosnd_2) to (dorcv_1);
      \draw[bend left=30,dashed,shorten <=-2pt,{Circle[fill=white]}-{Triangle[fill=white]}] (dosnd_1) to (dorcv_2);
      \draw[bend right=35,dashed,shorten <=-2pt,{Circle[fill=white]}-{Triangle[fill=white]}] (dosnd_2) to (dorcv_1);

      \foreach \u/\v/\b in {%
        snd1_c.north/snd_1.south/left,%
        rcv1_c.south/rcv_2.north/right,%
        snd2_c.south/snd_2.north/right,%
        rcv2_c.north/rcv_1.south/left%
      }
      \draw[bend \b=20,dashed,shorten <=-2pt,{Circle[fill=white]}-{Triangle[fill=white]}] (\u) to coordinate[pos=.7](msg)  (\v);
    \end{tikzpicture}\begin{tikzpicture}[scale=.6,baseline={(0,0,0)}]
      \draw[->] (0,0,0) -- (4,0,0) node[anchor=north,pos=1] (alice){Alice};
      \draw[->] (0,0,0) -- (0,4,0) node[anchor=west,pos=1] {Bob};
      \draw[->] (0,0,0) -- (0,0,4) node[anchor=north west,pos=.95] {\(\$\)};

      \foreach \x/\y/\z in {%
        1
        /3
        /3} {
        \foreach \u/\v/\w in  {%
          2
          /0
          /2
        }
        {
          \draw[opacity=.4,red,dashed](\x,\y,0) -- (\u,\v,0);
          \draw[opacity=.4,densely dotted](\x,\y,\z) -- (\x,\y,0);
          \draw[opacity=.4,densely dotted](\u,\v,\w) -- (\u,\v,0);

          \draw[opacity=.4,red,dashed](0,\y,\z) -- (0,\v,\w);
          \draw[opacity=.4,densely dotted](\x,\y,\z) -- (0,\y,\z);
          \draw[opacity=.4,densely dotted](\u,\v,\w) -- (0,\v,\w);

          \draw[opacity=.4,red,dashed](\x,0,\z) -- (\u,0,\w);
          \draw[opacity=.4,densely dotted](\x,\y,\z) -- (\x,0,\z);
          \draw[opacity=.4,densely dotted](\u,\v,\w) -- (\u,0,\w);

        }
      }

      \draw[thick,red]
      (1
      ,3
      ,3
      )
      --
      (2
      ,0
      ,2
      );
      \node (med) at (1.6,1.2,2.4) [circle,fill,inner sep=2pt,opacity=.6]{};
      \node  at (0,1.2,2.4) [circle,fill,inner sep=.5pt,opacity=1,gray]{};
      \node  at (1.6,0,2.4) [circle,fill,inner sep=.5pt,opacity=1,gray]{};
      \node  at (1.6,1.2,0) [circle,fill,inner sep=.5pt,opacity=1,gray]{};

      \draw[thick,red]
      (2
      ,0
      ,2
      )
      -- (med.center);

    \end{tikzpicture}
    \caption{The \emph{To work or not to work?} collaboration}
    \label{to-work-not-coll-med}
  \end{figure}
  Bob sketches a first plan to synchronize their schedules, depicted in Figure~\ref{fig:work-life-balance}.
  However, %
  Alice points out that this diagram is no good, %
  among others because 
  it mixes data and event-based decisions.
  Alice remarks that this can be avoided by throwing a coin %
  in the morning, to decide who goes to work. %
  But that leaves the question of who should throw the coin and communicate the result. 
  Alice suggests that their secretary \med could help them out %
  by rolling a 10-sided die 
  each morning and %
  notifying them about who is going to go to work that day.

  Alice comes up with the more elaborate process %
  shown in \figref{to-work-not-coll-med}, %
  which lets Bob and Alice work 60\% and 40\% of the days,  %
  resp.,
  because Alice is 50\% more efficient than Bob. %
  In fact, %
  every point on the solid red line\footnote{%
    The projections to the \(x\)-\(y\), \(y\)-\(z\), and \(x\)-\(z\) plane
    are rendered as dashed lines for clarity. 
  }
  on the right in \figref{to-work-not-coll-med} %
  correspond to a choice of odds for going to work;
  the given \bpmn model corresponds to the dot slightly off the middle. %
  The coordinates of each point describe the expected gain in personal enjoyment %
  and profits for the company, respectively. %
  Thus, %
  there is a choice of parameter for any suitable work life balance. %
  One might want to maximize for distance from the origin, %
  but questions of fairness can be addressed as well.%
  \reversemarginpar\thnote{%
    Homework for the interested reader \tt \^{}\_\^{} \\
    IW: add a ref to papers that go in this direction? (I should be able to dig one out)
  }\normalmarginpar %
  \ However, Alice's criterion for choosing the odds is equal contribution %
  of both co-founders to the company income. 
\end{iflong}

In game theoretic terminology, %
\med is taking the role of a common source of randomness %
that is independent of the state of the game %
and does not need to observe the actions of the players. %
The specific formal notion that we shall use is that of an %
\emph{autonomous correlation device} %
{\cite[Definition~2.1]{Solan2002CorEq}}.

\begin{definition}[Autonomous correlation device%
  ]
  \label{def:auto-corr-device}
  An \emph{autonomous correlation device} is %
  a family of pairs \(\mathcal{D} = \device{\plr}{n}\)  %
  (that is indexed over natural numbers \(n\in\mathbb{N}\)) %
  each of which consists of
  \begin{itemize}
  \item %
    a family of finite sets of \emph{signals} \(\sig{\plr}{n}\), %
    (additionally) indexed over players;
    \ and
  \item %
    a function \(\sigd[n]\) that maps lists of
    signal vectors \(\vect \m {n-1} \in \prod_{k=1}^{n-1}\sigs[k]\) to
    probability distributions %
    \(\sigd[n]\vect \m {n-1} \in \Delta(\sigs[n])\) 
    over the Cartesian product \(\sigs[n] = \prod_{\plr=1}^{\smash{|\N|}}\sig \plr n \)
    of all signal sets~\(\sig \plr n \). %
  \end{itemize}
\end{definition}
We shall refer to  operators of autonomous correlation devices as %
\emph{mediators}, which guide the actions of players during the game. %

Each correlation device for a game induces an extended game, 
which proceeds in \emph{stages}. %
\ifLong{Concerning the example of Bob and Alice, %
\figref{to-work-not-coll-med} is a description of the extended game of %
the collaboration that Bob drafted %
combined with the simplest possible correlation device -- %
assuming that each of the roles restarts at the end of the day.\par}  %
In general, %
given a game and an autonomous correlation device, %
the \(n\)-th stage begins with the mediator %
drawing a signal vector \(\m_n \in \sigs[n] = \prod_{\plr=1}^{\smash{|\N|}}\sig \plr n\) %
according to the device distribution~\(\sigd[n]\vect \m {n-1}\) -- e.g., \med rolling the die -- %
and sending the components to the respective players -- the sending of messages to Bob and Alice (in one order or the other). 
Then, each player~\(\plr\) chooses an available action~\(\act[\plr]_n\). %
This choice can be based on 
the respective component \(\m^{\plr}_n\) of the signal vector~\(\m_n\in \sigs[n]\), %
information about previous states \(\state[k]\) of the game~\(G\), %
and moves \(\act[{\plr[']}]_k\) of (other) players from the history.%
\footnote{%
  In the present paper, %
  we only consider games of perfect information, %
  which is suitable for business processes in a single organization or which are %
  monitored on a blockchain. %
}
After all players made their choice, %
we obtain an action profile~\(\act_n = \vectup{\act_n}{\smash{|\N|}}\). %
\ifLong{The reader may note, that the process model does not give Alive and Bob a lot of choice: once Mrs.\ Medina sent them the message with the decision, they can only notify each other (respectively receive this notification), and then work or surf / fish. However, in our game setting they can also choose the ``idle'' action, e.g., instead of going to work.}
\iwnote{Buffer overflow in parser. Der Satz ist 3x so lang wie er sein sollte\\
TH: habe drei Sätze draus gemacht.\\
IW: soweit gut. Hier kann man jetzt aber fragen, inwiefern Alice und Bob eine Wahl haben -- wenn Sie den Prozess ``implementiert'' haben. Weiss nicht, ob wir die Diskussion aufmachen wollen? Also: wenn Sie dem Prozess folgen, dann ist ja nach dem Würfelwurf alles automatisiert und es gibt jeweils nur eine Aktion oder ``idle''. Ist das dann die Auswahl: Prozess folgen oder ``idle''?\\
\ul{TH} genau richtig; müssen wir wohl doch 'nen Techreport schreiben ... 
\ul{IW} done
}

While playing the extended game described above, %
each player makes observations about the state %
and the actions of players; %
the role of the mediator is special %
insofar as it does not need and is also not expected to %
observe the run of the game. %
The ``local'' observations of each player are the basis of their strategies. %
\begin{definition}[Observation, strategy, strategy profile]
  \label{def:obs-n-strats}
  \ifLong{For a natural number \(n\in\mathbb{N}\), 
  an \emph{observation at stage~\(n\)} by player~\(\plr\)
  is %
  a tuple
  \(h=\langle \state[1],\m[\plr]_1,\act_1, \dotsc,\state[n-1],\m[\plr]_{n-1},\act_{n-1}, \state[n],\m[\plr]_n \rangle\)
    that
    consists of}%
  \ifShort{%
    An \emph{observation at stage~\(n\)} by player~\(\plr\)
    is %
  a tuple
  \(h=\langle \state[1],\m[\plr]_1,\act_1, \dotsc,\state[n-1],\m[\plr]_{n-1},\act_{n-1}, \state[n],\m[\plr]_n \rangle\)
  with
  }
    \begin{itemize}
  \item   one state~\(\state[k]\), signal~\(\m[\plr]_k\), and action profile~\(\act_k\), %
    for each number~\(k<n\), %
    
  \item 
    the \emph{current} state~\(\state[n]\), also denoted by \(\state[h]\), and
  \item  the \emph{current} signal \(\m[\plr]_n\). 
  \end{itemize}
  The set of all observations is denoted by %
  \(
  \obs{n}{\plr}\).
  The union \(\obs{}{\plr}= \bigcup_{n \in \mathbb{N}} \obs{n}{\plr}\) of %
   observations at any stage is the set of %
  \emph{observations} of player~\(\plr\).  %
  A \emph{strategy} is a map %
  \(\sigma^{\plr} \colon \obs{}{\plr} \to \Delta(\acts[\plr])\) %
  from observations to probability distributions %
  over actions that are available at the current state of histories, %
  i.e., %
  \(\sigma^{\plr}_{h}(\act[\plr]) =0\)
  if \(\act[\plr]\notin \acts[\plr](\state[h])\),
  for all histories~\(h \in \obs{}{\plr}\). %
  A \emph{strategy profile} is a player-indexed family of strategies %
  \(\{\sigma^{\plr}\}_{\plr\in\N}\).
\end{definition}
Thus, %
each of the players observes the history of other players, %
including the possibility of punishing other players for not heeding %
the advice of the mediator. %
This is possible %
since signals might give (indirect) information %
concerning the \mbox{(mis-)}\allowbreak{}behavior of players in the past, %
as remarked by Solan and Vieille \cite[p.~370]{Solan2002CorEq}: %
by revealing information about proposed actions of previous rounds, %
players can check for themselves whether %
some player has ignored some signal of the mediator.%

The data of a game, a correlation device, and a strategy profile induce %
probabilities for finite plays of the game,  
which in turn determine the expected utility of playing the strategy. %
Formally, %
an autonomous correlation device and a strategy profile with strategies for every player yield %
a probabilistic trajectory of a sequence
of ``global'' states, signal vectors of all players, and complete action profiles, %
 dubbed \emph{history}. %
The formal details are as follows. 
\begin{definition}[History and its probability]
  \label{def:hist-n-their-probs}
  \ifLong{Given a natural number \(n\in\mathbb{N}\), %
  a \emph{history at stage~\(n\)} is a tuple
  \(h = \langle \state[1],\m_1,\act_1, \dotsc,\state[n-1],\m_{n-1},\act_{n-1}, \state[n],\m_n \rangle\) 
  that consists of}%
  \ifShort{%
    A \emph{history at stage~\(n\)} is a tuple
  \(h = \langle \state[1],\m_1,\act_1, \dotsc,\state[n-1],\m_{n-1},\act_{n-1}, \state[n],\m_n \rangle\) 
  that consists of}
  \begin{itemize}
  \item one state~\(\state[k]\), signal vector~\(\m_k\), and action profile~\(\act_k\), %
    for each number~\(k<n\), %
  \item 
    the \emph{current} state~\(\state[n]\), often denoted by \(\state[h]\), and
  \item  the \emph{current} signal vector~\(\m_n\). 
  \end{itemize}
  The set of all histories at state~\(n\) is denoted by 
  \(
    \hist{n}\).
  The union \(\hist{} = \bigcup_{n\in\mathbb{N}}\hist{n}\) of histories at arbitrary stages is %
  the set of \emph{finite histories}. %
  The \emph{probability} of a finite history\ifLong{\footnote{%
    Note that these probabilities induce %
    the probability measure on infinite histories from %
    Solan and Vieille \cite{Solan2002CorEq}; %
    thus, we extend the notation of \textit{op.~cit.} %
    to finite histories. %
  }}\(h= \langle \state[1],\m_1,\act_1, \dotsc,\state[n-1],\m_{n-1},\act_{n-1}, \state[n],\m_n \rangle\) in the context of %
  a correlation device~\(\mathcal{D}\), %
  an initial state~\(\state\), and %
  a strategy profile~\(\sigma\) is defined %
  as follows, by recursion %
  over the length of histories. %
  \begin{description}
  \item[\(n=1\):]
    \(%
      \PP{\state}{\sigma}(\langle\state[1],\m_1\rangle) = %
      \begin{cases} 
        0
        & \text{if }\state\neq \state[1]
        \\
        \sigd[1]\langle\rangle(\m_1)
        & \text{otherwise}
      \end{cases}
    \)~\\[1ex]
  \item[\(n>1\):]
    {
      \newcommand{\nplusone}{n}
    \(%
      \PP{\state}{\sigma}(\langle \hbar,\act_{n-1},\state[\nplusone],\m_{\nplusone}\rangle%
      ) %
      =
      \underbrace{p_{\langle\hbar\rangle}(\act_{n-1})}_{\makebox[0pt]{\(\scriptstyle\prod_{\plr \in \N} \sigma^{\plr}_{\langle\hbar\rangle}(\act[\plr]_{n-1})\)}}
      \prob{\state[\nplusone]}{\state[n-1]}{\act_{n-1}}
      \underbrace{p_{\sigd[n-1]}(\m_{\nplusone})
       }_{\makebox[0pt]{\(\scriptstyle\sigd[n-1]\langle\m_1,\dotsc,\m_{n-1}\rangle(\m_{\nplusone})\)~~~}}
     \)}%

  \end{description}
\end{definition}
Again, %
note that the autonomous correlation device does not %
``inspect'' the states of a history, in the sense that %
the distributions over signal vectors \(\sigd[n]\) are %
\emph{not} parameterized over \emph{states} from the history, %
but only over previously drawn \emph{signal vectors} -- %
whence the name. 
\begin{definition}[Mean expected payoff]
  The \emph{mean expected payoff} of player~\(\plr\) for stage~\(n\)
  is
  \(
    \epay n (\mathcal{D},\state,\sigma) %
    = %
\sum_{h \in \obs {n+1} {}} \frac{\PP{\state}{\sigma}(h)}{n}\sum_{k=1}^n \pay[\plr](\state[k],\act_k)
\)
where \(h=\langle \state[1],\m_1,\act_1, \dotsc \act_{n}, \state[n+1],\m_{n+1} \rangle\). %
\end{definition}

At this point, 
we can address the question of what a good strategy profile is  %
and fill in all the details of the idea that an equilibrium is %
a strategy profile that does not give players any good reason to %
deviate unilaterally. %
We shall tip our hats to game theory and 
use the notation \((\strat[\plr{}]{},\sigma^{-\plr})\) %
for the strategy profile which is obtained by %
``overwriting'' the single strategy \(\sigma^{\plr}\) of player~\(\plr\) %
with a  strategy~\(\strat[\plr{}]{}\) (which might, but does not have to be different); %
thus, %
the expression `\((\strat[\plr{}]{},\sigma^{-\plr})\)' denotes the unique strategy subject to equations
\((\strat[\plr{}]{},\sigma^{-\plr})^{\plr} = \strat[\plr{}]{}\) and
\((\strat[\plr{}]{},\sigma^{-\plr})^{\plr[']} = \sigma^{\plr[']}\) %
(for~\(\plr\neq\plr[']\)). 

\begin{definition}[Autonomous correlated \(\boldsymbol\varepsilon\)-equilibrium]
  Given a positive real~\(\varepsilon >0\), %
  an \emph{autonomous correlated \(\varepsilon\)-equilibrium} %
  is a pair \(\langle \mathcal{D},{\sigma^*}\rangle\), %
  which consists of %
  an autonomous correlation device~\(\mathcal{D}\) %
  and a strategy profile~\({\sigma^*}\) for which %
  there exists a natural number \(n_0\in \mathbb{N}\) such that %
  for any alternative strategy~\(\sigma^{\plr}\) of any player~\(\plr\), %
  \begin{equation}
    \label{eq:eps-corr-eq-solan-vieille}
    \epay {n} (\mathcal{D},\state,{\sigma^*}) \geq
    \epay {n} \left(\mathcal{D},\state,(\sigma^{\plr},{\sigma^*}{}^{-\plr})\right) -\varepsilon
  \end{equation}
  Equation~\eqref{eq:eps-corr-eq-solan-vieille}
  holds, %
  for all \(n\geq n_0\) and all states~\(\state \in\states\).
\end{definition}
Thus, %
a strategy is an autonomous correlated \(\varepsilon\)-equilibrium %
if the benefits that one might reap in the long run %
by unilateral deviation from the strategy are negligible %
as \(\varepsilon\) can be arbitrarily small. %
In fact, %
other players will have ways to %
punish deviation from the equilibrium \cite[\S~3.2]{Solan2002CorEq}. %


\subsection{Petri nets and their operational semantics}
\label{sec:petri-nets}
Elementary net systems~\cite{RozenbergE96ENS} are expressive enough for %
the purposes of the present paper; %
however, %
for the sake of simplicity, %
we refer to them as Petri nets.%
\footnote{%
  In fact, %
  the results of the paper apply to the general case, %
  \textit{mutatis mutandis}. %
  Also, %
  note that we leave the flow relation implicit %
  as its identity is immaterial for %
  the execution semantics. %
}
\begin{definition}[Petri net, marking, and marked Petri net]
  A \emph{Petri net} is a tuple
  \(\Net = (\pl,\tr,\pre,\post)\) 
  that consists of
  \begin{itemize}
  \item a finite set of \emph{places} \(\pl\); 
  \item a finite set of \emph{transitions} \(\tr\); and 
  \item two functions \(\pre,\post \colon \tr \to \wp \pl\setminus\{\varnothing\}\) that %
    assign to each transition \(t\in\tr\) its %
    \emph{pre-set \(\pre[t] \subseteq \pl\)} and \emph{post-set \(\post[t] \subseteq \pl\)}, %
    respectively, which are both required to be non-empty.%
  \end{itemize}
  A \emph{marking} of a Petri net~\(\Net\) is a multiset of places~\(\M\), %
  i.e., a function \(\M \colon \pl \to \mathbb{N}\) that assigns to %
  each place \(p\in\pl\) a non-negative integer~\(\M(p) \geq 0\). %
  A \emph{marked Petri net} is %
  a tuple \(\Net = (\pl,\tr,\pre,\post,\M[0])\) whose %
  first four components \((\pl,\tr,\pre,\post)\) are a Petri net %
  and whose last component~\(\M[0]\) %
  is the \emph{initial marking}, %
  which is a marking of the latter Petri net. %
\end{definition}
One essential feature of Petri nets is 
the ability to execute several transitions concurrently -- 
possibly several occurrences of one and the same transition. %
However, %
we shall only encounter situations in which a set of transitions fires. %
Again, %
for the sake of simplicity, %
we shall use the general term \emph{step}. %
We fix a Petri net \(\Net = (\pl,\tr,\pre,\post)\) %
for the remainder of the section. %
\begin{definition}[Step, step transition, reachable marking]
  \label{def:step-trans-PN}
  A \emph{step} in the net~\(\Net\) is a set of transitions %
  \(\step \subseteq \tr\). %
  The \emph{transition relation} of a step \(\step \subseteq \tr\)
  relates a marking~\(\M\) to another marking~\(\M'\), %
  in symbols \(\M \fire{\step}\M[']\), %
  if the following two conditions are satisfied,
  for every place \(p \in \pl\). 
  \begin{enumerate}
  \item \(\M(p) \geq |\{t \in \step \mid p \in \pre[t]\}|\)
  \item \(\M['](p) = \M(p) - |\{t \in \step \mid p \in \pre[t]\}| +|\{t \in \step \mid p \in \post[t]\}|\)
  \end{enumerate}
  We write \(\M\fire{}\M[']\) if \(\M\fire{\step}\M[']\) holds for some step~\(\step\) %
  and denote the reflexive transitive closure of the relation~\(\fire{}\) by~\(\fire{}^*\). %
  A marking~\(\M[']\) is \emph{reachable} in a marked Petri net
  \(\Net = (\pl,\tr,\pre,\post,\M[0])\) %
  if \(\M[0]\fire{}^*\M[']\) holds (in the net \((\pl,\tr,\pre,\post)\)). %
\end{definition}
For a single transition~\(t \in \tr\), %
we write \(\M \fire{t}\M[']\) instead of \(\M \fire{\{t\}}\M[']\). %
Note that the empty step is always fireable, %
i.e., for each marking~\(m\), %
we have an “idle” step \(\M \fire{\varnothing} \M\). %

Recall that a marked Petri net \(\Net = (\pl,\tr,\pre,\post,\M[0])\) is \emph{safe} if
all reachable markings~\(\M[']\) have at most one token in any place, %
i.e., if they satisfy \(\M['](p)\leq 1\), %
for all \(p \in \pl\). %
Thus, a marking~\(\M\) corresponds to a set \(\hat\M\subseteq \pl\) %
satisfying \(p \in \hat\M\) iff \(\M(p)>0\); %
for convenience, %
we shall identity a safe marking~\(\M\) with its set of places~\(\hat\M\). %
The main focus will be on Petri nets that are %
safe and \emph{extended free choice}, %
i.e., if the pre-sets of two transitions have a place in common, %
the pre-sets coincide. %
Finally,
recall that the \emph{conflict relation}, denoted by~\(\#\), 
relates two transitions if their pre-sets intersect, %
i.e., \(t \mathrel{\#} t'\) if \(\pre[t] \cap \pre[t']\neq \varnothing\), %
for \(t,t'\in\tr\); %
for extended free choice nets, %
the conflict relation is an equivalence relation. %


\section{Incentive alignment}
\label{sec:incentive-alignment}
Soundness of business processes %
in the sense of Van der Aalst~\cite{Aalst97soundness}  %
implies termination if transitions are governed by 
a strongly fair scheduler~\cite{vanderAalst2011}\ifShort{; indeed, }%
\ifLong{: %
  \begin{quote}
    If we assume a strong notion of fairness,
    then the first requirement implies that %
    eventually state [o] is reached. Strong fairness, %
    sometimes also referred to as “impartial” or “recurrent” [KA99], means that %
    in every infinite firing sequence, each transition fires infinitely often. %
    Note that weaker notions of fairness are not sufficient, see Figure 2 in [KA99]. 
  \end{quote}
Indeed,
}%
such a scheduler fits the intra-organizational setting. %
\iwnote{bitte checken in long version - zitiere ich das richtig?}
However, %
as discussed for the order-to-cash process model, %
unfair scheduling practices could arise in the inter-organizational setting %
if undesired behavior yields higher profits. %
We consider incentive alignment to rule out %
scenarios that lure actors into counterproductive behavior.  %
We even can check whether %
all activities in a given \bpmn model with utility annotations %
are relevant and profitable.

As \bpmn models have established Petri net semantics~\cite{dijkman2008semantics}, %
it suffices to consider the latter for the game theoretic aspects of incentive alignment. %
As a preparatory step, %
we extend Petri nets with utility functions %
as pioneered by von Neumann and Morgenstern~\cite{morgenstern1953theory}. %
Then we describe two ways to associate a stochastic game %
to a Petri net with transition-based utilities: %
the first game retains the state space %
and the principal design choice concerns transition probabilities; %
the second game is the restarting
version of the first game, %
which will turn out to be better suited to analyze business processes %
and even gives rise to %
a tight connection with the soundness property %
of workflow nets~\cite{Aalst97soundness}. %
Finally, %
we define incentive alignment in full formal rigor based on stochastic games and %
show that %
the soundness property for workflows nets can be ``rediscovered'' as %
a special case of incentive alignment; %
in other words, %
soundness is conserved, and therefore \iwnote{Bitte checken ob du mit dem Halbsatz (in dieser Zeile) glücklich bist}
incentive alignment is a conservative extension of soundness. %

\subsection{Petri nets with utility and role annotations}
We assume that %
costs (respectively profits) are incurred (resp.\ gained) %
per task 
and that, in particular, utility functions do not depend on the state. %
Note that the game theoretic results do not require this assumption; %
however, %
this assumption does not only avoid clutter, but  %
also retains the spirit of the \textsc{abc} method~\cite{AMA3rdEdition} %
and is in line with the work of Herbert and Sharp~\cite{HerbertSharp12StochMC}. %
\begin{definition}[Petri net with transition payoffs and roles] 
  \label{def:petri-net-payoffs-roles}
  For a set of roles~\(\R\), %
  a \emph{Petri net with transition payoffs and roles} %
  is a triple \((\Net,\pay,\rl)\) where
  \begin{itemize}
  \item \(\Net = (\pl,\tr,\pre{},\post{},\M[0])\) is a marked Petri net with initial marking~\(\M[0]\),
  \item 
    \(\pay \colon \R \to T \to \mathbb{R}\) is a \emph{utility function}, and %
  \item \(\rl \colon T \rightharpoonup \R\) is %
    a partial
		function, assigning at most one role to each transition. 
  \end{itemize}
  The \emph{utility} \(\pay[\plr](\step)\) of a step~\(\step \subseteq T\) is %
  the sum of the utilities of its elements, %
  i.e., \(\pay[\plr](\step) = \sum_{t \in \step} \pay[\plr](t)\), %
  for each role \(\plr\in\R\). 
\end{definition}
As a consequence of the definition, %
the idle step has zero utility. %
We have included the possibility that some of the transitions are %
not controlled by any of the roles (of a \bpmn model) %
by using a partial function from transitions to roles; %
we take a leaf out of the game theorist's book %
and attribute the missing role to \emph{nature}. 

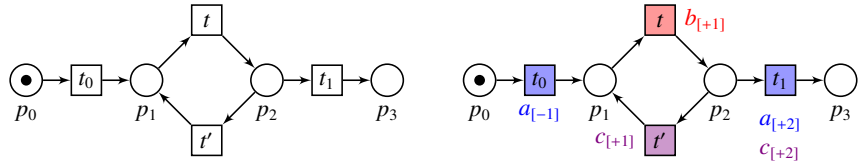
\begin{figure}[tb]
  \centering
  \tikzstyle{place}+=[minimum size=3ex]
  \begin{tikzpicture}[>=latex',semithick,scale=.8,baseline={(p0.base)}]
  \node[place,tokens=1,label=below:{\(p_0\)}] (p0) at (0,0){};
  \node[place,label=below:{\(p_1\)}] (p1) at (2,0) {};
  \node[place,label=below:{\(p_2\)}] (p2) at (4,0) {};
  \node[place,label=below:{\(p_3\)}] (p3) at (6,0) {};

  \node[transition
  ] (t0) at (1,0) {{\(t_0\)}};

\node[transition
] (t) at (3,1) {{\(t\)}};
\node[transition
] (t') at (3,-1) {\(t'\)};
\node[transition
] (t1) at (5,0) {{\(t_1\)}};
  \draw[->] (p0) -- (t0);
  \draw[->] (t0) -- (p1);
  \foreach \p/\t/\q in {p1/t/p2,p2/t'/p1,p2/t1/p3}
  {
    \draw[->] (\p) --(\t);
    \draw[->] (\t) --(\q);
  };

\end{tikzpicture}
\qquad
  \begin{tikzpicture}[>=latex',semithick,scale=.8,baseline={(p0.base)}]
  \node[place,tokens=1,label=below:{\(p_0\)}] (p0) at (0,0){};
  \node[place,label=below:{\(p_1\)}] (p1) at (2,0) {};
  \node[place,label=below:{\(p_2\)}] (p2) at (4,0) {};
  \node[place,label=below:{\(p_3\)}] (p3) at (6,0) {};

  \node[transition,label=below:{\color{blue}\(a_{[-1]}\)},fill=blue!40!white] (t0) at (1,0) {{\(t_0\)}};

  \node[transition,label=right:{\color{red}\(b_{[+1]}\)},fill=red!40!white] (t) at (3,1) {{\(t\)}};
  \node[transition,label=left:{\color{violet}$c_{[+1]}$},fill=violet!40!white] (t') at (3,-1) {\(t'\)};
  \node[transition,%
  label=below:{%
    \begin{tabular}[t]{@{}l@{}}
      \color{blue}$a_{[+2]}$\\%
      \color{violet}$c_{[+2]}$%
    \end{tabular}
  },%
  fill=blue!40!white] (t1) at (5,0) {{\(t_1\)}};

  \draw[->] (p0) -- (t0);
  \draw[->] (t0) -- (p1);
  \foreach \p/\t/\q in {p1/t/p2,p2/t'/p1,p2/t1/p3}
  {
    \draw[->] (\p) --(\t);
    \draw[->] (\t) --(\q);
  };

\end{tikzpicture}
  \caption{Extending Petri nets with role and utility annotations}
  \label{fig:ext-pn}
\end{figure}
\figref{fig:ext-pn} displays a  %
Petri net on the left. %
The names of the places $p_1, \ldots, p_4$ will be convenient later. 
In the same figure on the right, %
we have added annotations that carry information
concerning roles, costs, and profits %
in the form of lists of role-utility pairs next to transitions. %
E.g.,
the transition $t_0$ is assigned to role~$a$ and %
firing \(t_0\) results in utility~$-1$ for~\(a\),
i.e., one unit of cost. %
The first role in each list %
denotes responsibility for the transition %
and we have omitted entries with zero utility. %
We also have colored transitions with the same color as the role assigned to it.


There are natural translations from \bpmn models with %
payoff annotations for activities to %
Petri nets with payoffs and roles %
(relative to any of the established Petri net semantics for %
models in \bpmn~\cite{dijkman2008semantics}). %
If pools are used,
we take one role per pool and %
each task is assigned to its enclosing pool; %
for pairs of sending and receiving tasks or events, %
the sender is responsible for the transition to be taken. %
The only subtle point concerns the role of nature. %
When should we blame nature for the data on which choices are based? %
The answer depends on the application at hand. %
For instance, %
let us consider the order-to-cash model of~\figref{fig:o2c-basic}: %
whether or not the supplier is out of stock is %
only partially within the control of supplier; %
also, %
we can avoid the need to meticulously model %
all factors that might lead to damage of a shipped good %
by blaming nature. %
In a first approximation, %
we simply let nature determine the state of the stock and %
damage goods at will. %
\thnote{%
  ... or something along these lines.%
}

\subsection{Single process instances and the base game with fair conflicts}
We now describe how each Petri net with transition payoffs and roles %
gives rise to a stochastic game, based on two design choices: %
each role can execute only one (enabled) transition at a time and %
conflicts are resolved in a probabilistically fair manner. 
For example, %
for the net on the right in \figref{fig:ext-pn}, %
we take four states \(p_0,p_1,p_2,p_3\), %
one for each reachable marking. %
The Petri net does not prescribe %
what should happen if roles~\(a\) and~\(c\) both try to %
fire transitions~\(t_1\) and~\(t'\) simultaneously %
if the game is in state~\(p_2\). %
The simplest probabilistically fair solution consists of flipping a coin; %
depending on the outcome, %
the game continues in state~\(p_1\) or in state~\(p_3\). %
For the general case,
let us fix a safe, extended free-choice net \((\Net,\pay,\rl)\) with payoffs and roles %
whose initial marking is~\(\M[0]\). %
\begin{definition}[The base game with fair conflicts]
  \label{def:basic-game}
  Let \(\Xs\subseteq \wp\tr\) be the partitioning of the set of transitions into equivalence classes of %
  the conflict relation on the set of transitions,  %
  i.e., \(\Xs=\{ \{t' \in \tr \mid t' \mathrel{\#} t \} \mid t \in\tr\}\); %
  its members are called \emph{conflict sets}.
  Given a safe marking~\(\M\subseteq\pl\) and a step~\(\step\subseteq\tr\), %
  a \emph{maximal \(\M\)-enabled sub-step} is a step \(\step[']\) %
  that is enabled at the marking~\(\M\), 
  is contained in the step~\(\step\), and %
  contains one transition of each conflict set that has a non-empty intersection with the step, %
  i.e., such that all three of %
  \(\M\fire{\step[']}\),  %
  \(\step['] \subseteq \step\) and %
  \(|\step[']| = |\{ \X \in \Xs \mid \step \cap \X \neq \varnothing\}|\) hold. %
  We write \(\step['] \sqsubseteq_{\M} \step\) if 
  the step~\(\step[']\) is a maximal \(\M\)-enabled sub-step of the step~\(\step\). %

  The \emph{base game with fair conflicts} \(\game\) of the net \((\Net,\pay,\rl)\) is defined as follows. 
  \begin{itemize}
  \item %
    The set of players \(\N \defeq \R \cup \{\bot\}\) is the set of roles and %
		\emph{nature},~\(\bot\notin \R\). \iwnote{no longer so fresh}
  \item %
    The state space \(\states\) is the set of reachable markings, %
    i.e., \(\states = \{\M[']\mid \M[0]\fire{}^*\M[']\}\). 
  \item %
    The action set of an individual player~\(\plr\) is
    \(\acts[\plr] \defeq \{\varnothing\} \cup \{ \{t\} \mid t \in \tr, \rl(t) = \plr\}\), %
    which consists of the empty set and possibly singletons of transitions, %
    where \(\rl(t) =\bot\) if \(\rl(t)\) is not defined. %
    We identify an action profile \(\act \in \acts = \prod_{\plr=1}^{|\N|}\acts[\plr]\) with
    the union of its components \(\act\equiv\bigcup_{\plr\in\N} \act[\plr]\).  %
  \item %
    In a given state $\M$, the available actions of player~\(\plr\) are the enabled transitions, %
    i.e.,
    \(\acts[\plr](\M) = \{ \{t\} \in\acts[\plr] \mid \M\fire{t}\}\).
  \item%
    \(%
    \prob{\M[']}{\M}{\step} =
    \sum_{{
        \step[']\sqsubseteq_{\M}\step \text{ s.t. }
        \M\fire{\step[']}\M[']
      }}
    \prod_{{\X\in\Xs\text{ s.t. } \step\cap\X\neq\varnothing}}\frac{1}{
      |\step \cap \X|}
    \) 
  \item %
    \(\pay[\plr](\M,\step) = \sum_{t\in\step}\pay[\plr](t)\) %
    if \(\plr\in\R\) and \(\pay[\bot](\M,\step) = 0\), %
    for all \(\step \subseteq \tr\), and \(\M\subseteq\pl\).  
  \end{itemize}
\end{definition}
Let us summarize the stochastic game of %
a given Petri net with transition payoffs and roles. %
The stochastic game has the same state space as the Petri net, %
i.e., the set of reachable markings.
The available actions for each player at a given marking %
are the enabled transitions that are assigned to the player, %
plus the ``idle'' step. 
Each step  comes with a state-independent payoff, %
which sums up the utilities of each single transition,  
for each player~\(\plr\). %
In particular, if all players chose to idle, %
the corresponding action profile is the empty step~\(\varnothing\), %
which gives \(0\)~payoff. %
The transition probabilities implement the idea that %
all transitions of an action profile get a fair chance to fire, %
even if the step contains conflicting transitions. %
Let us highlight the following two points for a fixed marking and step: %
(1)~given a maximal enabled sub-step, %
we roll a fair ``die'' for each conflict set %
where the ``die'' has one ``side'' for each transition in the conflict set that %
also belongs to the sub-step %
(unless the ``die'' has zero sides); %
(2)~there might be several choices of maximal enabled sub-steps %
that lead to the same marking. %
In the definition of transition probabilities,  %
the second point is captured by %
summation over maximal enabled sub-steps of the step %
and the first point corresponds to a product of probabilities for %
each outcome of ``rolling''  one of the ``dice''. %

We want to emphasize that if additional information about %
transition probabilities are known, it should be incorporated. %
In a similar vein, %
one can adapt the approach of Herbert and Sharp~\cite{HerbertSharp12StochMC}, %
which extends the \bpmn language with %
probability annotations for choices. %
However, 
as we are mainly interested in \textit{a priori} analysis, %
our approach might be preferable since it avoids arbitrary parameter guessing. %
The most important design choice that we have made concerns the role of {nature}, %
which we consider as absolutely neutral;
it is not even  concerned with progress of the system %
as it does not benefit from transitions being fired. %

Now, %
let us consider once more  the order-to-cash process. %
If the process reaches the state in which customer's next step is payment, %
there is no incentive for paying. %
Instead, customer can choose to idle, \emph{ad infinitum}. %
In fact, %
this strategy yields maximum payoff for the customer. %
The \bpmn-model does not give any means for punishing %
customer's payment inertia. %
However, %
there is not even any incentive for shipper to %
pick up the goods! %
Incentives in the single instance scenario can be fixed, 
e.g., by adding escrow. 
However, %
in the present paper, %
we shall give yet a different perspective: %
we repeat the process indefinitely. 




\subsection{Restarting the game for multiple process instances}
The single instance game from Definition~\ref{def:basic-game} has one major drawback. %
It allows to analyze only a single instance of a business process. %
We shall now consider a variation of the stochastic game, %
which addresses the case of multiple instances in the simplest form. 
The idea is the same as the one for looping versions of workflow nets %
that have been considered in the literature, %
e.g., to relate soundness with liveness \cite[Lemma~5.1]{vanderAalst2011}: %
we simply restart the game in the initial state %
whenever we reach a final marking. %
\begin{definition}[Restart game]%
  \label{def:restart-game}%
    %
  A safe marking \(\M\subseteq \pl\) is \emph{final} if %
  it does not intersect with any pre-set, %
  i.e., if \(\M \cap \pre[t] = \varnothing \), %
  for all transitions~\(t \in \tr\); %
  we write \(\M\downarrow\) if the marking~\(\M\) is final, %
  and \(\M\not\,\downarrow\) if not. %
  Let  \(\game\)  be the {base game with fair conflicts}  of the net \((\Net,\pay,\rl)\). %
  The \emph{restart game} of the net \((\Net,\pay,\rl)\)
  is the game \(\rgame\) with %
  \begin{itemize}
  \item \(\rstates = \states \setminus \{\M['']\subseteq \pl\mid \M['']\downarrow\}\);
  \item \(\rprob{\M[']}{\M}{\step} =
    \begin{cases}
      \prob{\M[']}{\M}{\step} &\text{if } \M['] \neq \M[0] 
      \\
      \prob{\M[0]}{\M}{\step} + \sum_{\M['']\downarrow} \prob{\M['']}{\M}{\step} &\text{if } \M[']=\M[0] 
    \end{cases}
    \) 
  \end{itemize}
  for all \(\M,\M[']\in\rstates\); and
  the available actions restricted to \(\rstates \subseteq \states\), 
  i.e., %
  \(\racts[\plr](\state) = \acts[\plr](\state)\), %
  for \(\state \in\rstates\). %
\end{definition}
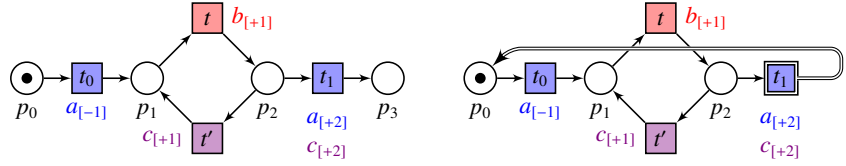
\begin{figure}[tb]
  \centering
    \tikzstyle{place}+=[minimum size=3ex]
    \begin{tikzpicture}[>=latex',semithick,scale=.8]
  \node[place,tokens=1,label=below:{\(p_0\)}] (p0) at (0,0){};
  \node[place,label=below:{\(p_1\)}] (p1) at (2,0) {};
  \node[place,label=below:{\(p_2\)}] (p2) at (4,0) {};
  \node[place,label=below:{\(p_3\)}] (p3) at (6,0) {};

  \node[transition,label=below:{\color{blue}\(a_{[-1]}\)},fill=blue!40!white] (t0) at (1,0) {{\(t_0\)}};

  \node[transition,label=right:{\color{red}\(b_{[+1]}\)},fill=red!40!white] (t) at (3,1) {{\(t\)}};
  \node[transition,label=left:{\color{violet}$c_{[+1]}$},fill=violet!40!white] (t') at (3,-1) {\(t'\)};
  \node[transition,%
  label=below:{%
    \begin{tabular}[t]{@{}l@{}}
      \color{blue}$a_{[+2]}$\\%
      \color{violet}$c_{[+2]}$%
    \end{tabular}
  },%
  fill=blue!40!white] (t1) at (5,0) {{\(t_1\)}};

  \draw[->] (p0) -- (t0);
  \draw[->] (t0) -- (p1);
  \foreach \p/\t/\q in {p1/t/p2,p2/t'/p1,p2/t1/p3}
  {
    \draw[->] (\p) --(\t);
    \draw[->] (\t) --(\q);
  };

\end{tikzpicture}
\qquad
  \begin{tikzpicture}[>=latex',semithick,scale=.8]
  \node[place,tokens=1,label=below:{\(p_0\)}] (p0) at (0,0){};
  \node[place,label=below:{\(p_1\)}] (p1) at (2,0) {};
  \node[place,label=below:{\(p_2\)}] (p2) at (4,0) {};
  \node[place,
  draw=none] (p3) at (6,0) {};

  \node[transition,label=below:{\color{blue}\(a_{[-1]}\)},fill=blue!40!white] (t0) at (1,0) {{\(t_0\)}};

  \node[transition,label=right:{\color{red}\(b_{[+1]}\)},fill=red!40!white] (t) at (3,1) {{\(t\)}};
  \node[transition,label=left:{\color{violet}$c_{[+1]}$},fill=violet!40!white] (t') at (3,-1) {\(t'\)};
  \node[transition,%
  label=below:{%
    \begin{tabular}[t]{@{}l@{}}
      \color{blue}$a_{[+2]}$\\%
      \color{violet}$c_{[+2]}$%
    \end{tabular}
  },%
  fill=blue!40!white,double] (t1) at (5,0) {{\(t_1\)}};

  \draw[->] (p0) -- (t0);
  \draw[->] (t0) -- (p1);
  \foreach \p/\t/\q in {p1/t/p2,p2/t'/p1
  }
  {
    \draw[->] (\p) --(\t);
    \draw[->] (\t) --(\q);
  };


  \draw[->] (p2) -- (t1) ;
    \draw[rounded corners,->,ultra thick,white] (t1) -- (p3.center) -- ([yshift=1.5ex]p3.north) -- ([yshift=1.5ex,xshift=2.5ex]p0.north) -- (p0.60);

  \draw[rounded corners,->,thin,double] (t1) -- (p3.center) -- ([yshift=1.5ex]p3.north) -- ([yshift=1.5ex,xshift=3.5ex]p0.north) -- (p0.north east);
\end{tikzpicture} 
  \caption{Restarting process example}
  \label{fig:restart-game-ex}
\end{figure}
For workflow nets, %
the variation amounts to identifying the final place with the initial place. 
The passage to the restart game is illustrated in \figref{fig:restart-game-ex}. 
Note that the restart game of our example is drastically different.  %
Player~\(c\) will be better off ``cooperating'' %
and never choosing the action~\(t'\), %
but instead idly reaping benefits %
by letting players~\(a\) and~\(b\) do the work. %

\subsection{Incentive alignment w.r.t.\ proper completion and full liveness}
We now formalize the idea that participants want to expect benefits 
from taking part in a collaboration %
if agents behave rationally -- %
the standard assumption of game theory. %
The proposed definition of incentive alignment is %
in principle of qualitative nature, %
but it hinges on quantitative information, %
namely the expected utility for %
each of the business partners of an inter\-organizational process. %


Let us consider a Petri net with payoffs \((\Net,\pay,\rl)\), %
e.g., the Petri net semantics of a \bpmn model. %
Incentive alignment %
amounts to existence of equilibrium strategies %
in the associated restart game~\(\rgame\) %
(as per Definition~\ref{def:restart-game}) %
that eventually will lead to positive utility for every participating player. %
The full details are as follows. %
\begin{definition}[Incentive alignment w.r.t.\ completion and full liveness]
  Given an autonomous correlation device \(\mathcal{D}\), %
  a correlated strategy profile~\(\sigma\) is \emph{eventually positive} %
  if there exists a natural number \(\nnull\in\mathbb{N}\) such that, %
  for all larger natural numbers \(n>\nnull\), %
  the expected payoff 
  of every player is positive, %
  i.e.,   for all \(\plr\in \N\), %
  \(\epay n (\mathcal{D},\M[0],\sigma)>0\). %
  Incentives in the net \((\Net,\pay,\rl)\) %
  are \emph{aligned} with
  \begin{itemize}
  \item \emph{proper completion} if, %
    for every positive real~\(\varepsilon>0\), %
    there exist  %
    an autonomous correlation device~\(\mathcal{D}\) and %
    an eventually positive correlated \(\varepsilon\)-equilibrium strategy profile~\(\sigma\) %
    of the restart game~\(\rgame\) such that, %
    for every natural number~\(\nnull\in\mathbb{N}\), %
    there exists a history \(h\in \hist{n}\) at stage~\(n> \nnull\) with %
    current state \(\state[h]=\M[0]\) that has %
    non-zero probability, i.e., \(\PP{\M[0]}{\sigma}(h) >0\);
  \item 
    \emph{full liveness} if, %
    for every positive real~\(\varepsilon>0\), %
    there exist
    an autonomous correlation device~\(\mathcal{D}\) and %
    an eventually positive correlated \(\varepsilon\)-equilibrium strategy profile~\(\sigma\) %
    of the restart game~\(\rgame\) %
    such that,
    for every transition~\(t \in \tr\),
    for every reachable marking~\(\M[']\), and %
    for every natural number~\(\nnull\in\mathbb{N}\), %
    there exists a history %
    \[h=\langle \M['],\m_1,\act_1, \dotsc,\state[n-1],\m_{n-1},\act_{n-1}, \state[n],\m_n \rangle \in \hist{n}
    \]
    at stage~\(n> \nnull\) 
    with %
    \(t \in \act_{n-1}\) and 
    \(\PP{\M[']}{\sigma}(h) >0\). %
  \end{itemize}
\end{definition}
Both variations of incentive alignment ensure that %
all participants can expect to gain profits on average, eventually;
moreover, 
something ``good'' will always be possible in the future %
where something ``good'' is either restart of the game (upon completion) or
additional occurrences of every transition. %

There are several interesting consequences. %
First,
incentive alignment w.r.t.\ full liveness implies %
incentive alignment
w.r.t.\ proper completion, %
for the case of safe conflict-free Petri nets %
where the initial marking is only reachable via the empty transition sequence,
including the following class of workflow nets.
\begin{definition}[Elementary workflow net]
  An \emph{elementary workflow net} is
  an elementary net system \(\Net = (\pl,\tr,\pre,\post)\)~\cite{RozenbergE96ENS}
  such that 
  \begin{itemize}
  \item %
    there exist unique places \({}i\in\pl\) and \({}o\in\pl\),
    such that \({}i\notin \post[t]\) and \({}o\notin \pre[t]\), %
    for all \(t\in\tr\), %
    which are called the \emph{initial} and \emph{final} place, resp., and %
  \item %
    each place is in the pre- or post-set of some transition, %
    i.e., \(\pl = \bigcup_{t\in \tr} (\pre[t]\cup\post[t])\).
  \end{itemize}
\end{definition}
\reversemarginpar\thnote{%
  simply because the first transition taken will occur again, %
which in turn means that the process with restart\\
IW: Yeah, I'm not getting that. Also, I'm not getting how proper completion works without saying ``for any reachable marking'' in some form or other. I guess, I cannot just stay in the initial marking and endlessly execute the idle actions, because that strategy is not positive?\\
TH: Exactly! 100\% \\
IW: Please add the explanation to the TR.
}\normalmarginpar%
Next, %
note that incentive alignment w.r.t.\ full liveness
implies the soundness property for safe elementary workflow nets.\iwnote{Brauchen wir nicht auch free-choice?} %
The main insight is that correlated equilibria can serve %
a very special case of strongly fair schedulers, %
not only for the case of a single player. %
However, %
we can even obtain a characterization of soundness %
in terms of incentive alignment w.r.t. full liveness, %
which makes precise in which sense we have extended soundness conservatively. %
\begin{theorem}[Characterization of soundness]
  \label{thm:da-thm}
    Let \(\Net\) be %
  a elementary  workflow net that is safe and 
  extended free choice;
  let \((\Net,\rl \colon T \to \{\Sigma\},\underline 1)\)
  be the net with transition payoffs and roles where %
  \(\Sigma\) is a unique role,
  \(\rl\colon T \to \{\Sigma\}\) is the unique total function, %
  and \(\underline 1\) is the constant utility-1 function. %
  The soundness property holds in \(\Net\) if, and only if, %
  incentives in \((\Net,\rl \colon T \to \{\Sigma\},\underline 1)\)  %
  are aligned w.r.t.\ full liveness. %

\end{theorem}
The proof can be found in 
\ifShort{the extended version \cite[Appendix A]{extendedVersion}}%
\ifLong{Appendix~\ref{apx:proof}}%
.

Finally, %
the reader may wonder why we consider the restarting game. %
Indeed, for Petri nets that do not have any cycles, %
one could formalize the idea of incentive alignment 
using finite extensive form games 
for which correlated equilibria have been studied
as well \cite{Stengel2008ExtensiveFormCE}.
However, %
%
%
this alternative approach is only natural for
Petri nets with transition payoffs and roles, but \emph{without cycles}. 
%
In the present paper we have opted for a general approach,
which does not impose the rather strong restriction on nets to be acyclic.
Also let us emphasize that the restart \emph{games} are merely a means to an end.
We generate them from \emph{arbitrary} free-choice safe Petri nets with transition payoffs and roles,
which are the main objects of the theoretical result.

\section{Conclusions and future work}
\label{sec:concl}

We have described a game theoretic perspective on %
incentive alignment of inter\hyp{}organizational business processes. %
It applies to \bpmn collaboration models %
that have annotations for activity-based utilities for all roles. 
The main theoretical result is that %
incentive alignment is a conservative extension of %
the soundness property, %
which means that we have described a uniform framework
that applies the same principles to %
intra- and inter-organizational business processes. %
We have illustrated incentive alignment for %
the example of the order-to-cash process and an additional example %
that is tailored to illustrate the game theoretic element of mediators. %


The natural next step is the implementation of a tool chain %
that takes a \bpmn collaboration model with annotations, %
transforms it into a Petri net with transition payoffs and roles, %
which in turn is analyzed concerning incentive alignment, %
e.g., using algorithms for solving stochastic games~\cite{MacDermedIsbell2011AIII}. %
\ifLong{One might even consider to extend \textsc{prism}~\cite{kwiatkowska2016prism}
or the model checker \textsc{storm} \cite{DBLP:journals/corr/abs-2002-07080}.} 
A very challenging venue for future theoretical work is
the extension to the analysis of interleaved execution of %
several instances of a process.




\bibliographystyle{splncs04etal}
\bibliography{bibliography}

\begin{thebibliography}{10}
\providecommand{\url}[1]{\texttt{#1}}
\providecommand{\urlprefix}{URL }
\providecommand{\doi}[1]{https://doi.org/#1}

\bibitem{vanderAalst2011}
Van~der Aalst, W.M.P., Van~Hee, K.M., ter Hofstede, A.H.M., et~al.: Soundness
  of workflow nets: classification, decidability, and analysis. Formal Asp.
  Comp.  \textbf{23}(3),  333--363 (2011)

\bibitem{Aalst97soundness}
Van~der Aalst, W.M.P.: Verification of workflow nets. In: Application and
  Theory of Petri Nets, {ICATPN} (1997)

\bibitem{Aumann1974CorEq}
Aumann, R.J.: Subjectivity and correlation in randomized strategies. Journal of
  Mathematical Economics  \textbf{1}(1),  67--96 (1974)

\bibitem{cachon2004game}
Cachon, G.P., Netessine, S.: Handbook of Quantitative Supply Chain Analysis,
  chap. Game Theory in Supply Chain Analysis, pp. 13--65. Springer (2004)

\bibitem{Cartelli:2014}
Cartelli, V., Di~Modica, G., Manni, D., Tomarchio, O.: A cost-object model for
  activity based costing simulation of business processes. In: European
  Modelling Symposium (2014)

\bibitem{dijkman2008semantics}
Dijkman, R.M., Dumas, M., Ouyang, C.: Semantics and analysis of business
  process models in bpmn. Information and Software technology  \textbf{50}(12),
   1281--1294 (2008)

\bibitem{Fundamentals-book}
Dumas, M., Rosa, M.L., Mendling, J., Reijers, H.A.: Fundamentals of Business
  Process Management. Springer, 2nd edn. (2018)

\bibitem{DBLP:journals/corr/abs-2002-07080}
Hensel, C., Junges, S., Katoen, J., Quatmann, T., Volk, M.: The probabilistic
  model checker storm. CoRR  \textbf{abs/2002.07080} (2020),
  \url{https://arxiv.org/abs/2002.07080}

\bibitem{HerbertSharp12StochMC}
Herbert, L., Sharp, R.: Using stochastic model checking to provision complex
  business services. In: IEEE Intl. Symp. High-Assurance Systems Engineering
  (2012)

\bibitem{correlatedEqu87Aumann}
J.~Aumann, R.: Correlated equilibrium as an expression of {Bayesian}
  rationality. Econometrica  \textbf{55}(1),  1--18 (1987)

\bibitem{Jaskiewicz2017handbook}
Ja{\'{s}}kiewicz, A., Nowak, A.S.: Non-Zero-Sum Stochastic Games, pp. 1--64.
  Springer (2017)

\bibitem{AMA3rdEdition}
Kaplan, R., Atkinson, A.: Advanced Management Accounting. Prentice Hall, 3rd
  edn. (1998)

\bibitem{kwiatkowska2018automated}
Kwiatkowska, M., Norman, G., Parker, D., Santos, G.: Automated verification of
  concurrent stochastic games. In: Intl. Conf. Quantitative Evaluation of
  Systems. pp. 223--239 (2018)

\bibitem{kwiatkowska2016prism}
Kwiatkowska, M., Parker, D., Wiltsche, C.: {PRISM}-games 2.0: A tool for
  multi-objective strategy synthesis for stochastic games. In: International
  Conference on Tools and Algorithms for the Construction and Analysis of
  Systems (2016)

\bibitem{leyton2008essentials}
Leyton-Brown, K., Shoham, Y.: Essentials of game theory: A concise
  multidisciplinary introduction. Synthesis lectures on AI and ML
  \textbf{2}(1),  1--88 (2008)

\bibitem{MacDermedIsbell2009NIPS}
MacDermed, L., Isbell, C.L.: Solving stochastic games. In: Conf. Neural
  Information Processing Systems. pp. 1186--1194 (2009)

\bibitem{MacDermedIsbell2011AIII}
MacDermed, L., Narayan, K.S., Isbell, C.L., Weiss, L.: Quick polytope
  approximation of all correlated equilibria in stochastic games. In: {AAAI}
  Conference (2011)

\bibitem{2018-Mendling-TMIS}
Mendling, J., Weber, I., Van~der Aalst, W.M.P., et~al.: Blockchains for
  business process management -- challenges and opportunities. ACM Transactions
  on Management Information Systems (TMIS)  \textbf{9}(1),  4:1--4:16 (Feb
  2018)

\bibitem{morgenstern1953theory}
Morgenstern, O., von Neumann, J.: Theory of games and economic behavior.
  Princeton university press (1953)

\bibitem{NarayananAnanth2004SupplyChainIA}
Narayanan, V., Raman, A.: Aligning incentives in supply chains. Harvard
  Business Review  \textbf{82},  94--102, 149 (12 2004)

\bibitem{RubinsteinOsborne1994book}
Osborne, M.J., Rubinstein, A.: A course in game theory. MIT press (1994)

\bibitem{PrasadEtAl15NashEqu}
Prasad, H., L., P., Bhatnagar, S.: Two-timescale algorithms for learning {Nash}
  equilibria in general-sum stochastic games. In: Intl. Conf. Auton. Agents and
  Multiagent Systems (2015)

\bibitem{RozenbergE96ENS}
Rozenberg, G., Engelfriet, J.: Elementary net systems. In: Lectures on Petri
  Nets {I:} Basic Models, Advances in Petri Nets, held in Dagstuhl. pp. 12--121
  (Sep 1996)

\bibitem{Shapley1095StochGames}
Shapley, L.S.: Stochastic games. Proc. Nat. Academy of Sciences
  \textbf{39}(10),  1095--1100 (1953)

\bibitem{Solan2002CorEq}
Solan, E., Vieille, N.: Correlated equilibrium in stochastic games. Games and
  Economic Behavior  \textbf{38}(2),  362--399 (2002)

\bibitem{SolanVieille2015PNAS}
Solan, E., Vieille, N.: Stochastic games. Proceedings of the National Academy
  of Sciences  \textbf{112}(45),  13743--13746 (2015)

\bibitem{Stengel2008ExtensiveFormCE}
von Stengel, B., Forges, F.: Extensive-form correlated equilibrium: Definition
  and computational complexity. Math. Oper. Res.  \textbf{33},  1002--1022
  (2008)

\bibitem{2016-Weber-BPM}
Weber, I., Xu, X., Riveret, R., et~al.: Untrusted business process monitoring
  and execution using blockchain. In: BPM (2016)

\bibitem{Weske-book}
Weske, M.: Business Process Management -- Concepts, Languages, Architectures.
  Springer, 3rd edn. (2019)

\bibitem{2019-Blockchain-Book}
Xu, X., Weber, I., Staples, M.: Architecture for Blockchain Applications.
  Springer (2019)

\end{thebibliography}
\appendix
\ifLong{%
  \newcommand{\apx}{\section}
\clearpage\apx{The conservative extension theorem}%
\label{apx:proof}%
\newcommand{\thefin}{[o]}
\newcommand{\theini}{[i]}

Let us quickly recall the definition of soundness. According to \cite{vanderAalst2011}\footnote{Emphasis is taken from the source.}, an (elementary) workflow net is
\begin{quote}
  sound if and only if the following three requirements are satisfied:
(1)~\emph{option to complete}:
for each case it is always still possible to reach the state which just marks
place \emph{end}, (2)~\emph{proper completion}:
if place \emph{end} is marked all other places are empty for a given case, and
(3)~\emph{no dead transitions}: it should be possible to execute an arbitrary activity by following the appropriate route
\end{quote}
where \emph{each case} means each reachable marking, %
\emph{state} means marking, 
\emph{marked} means marked by a reachable marking, %
\emph{activity} means \emph{transition}, and %
\emph{following the appropriate route} means
after executing the appropriate firing sequence
(see also \cite[Definition~5.1]{vanderAalst2011}).



\begin{lemma}[Proper completion from  safety and option to complete]
  \label{lem:proper-completion}
  If an elementary workflow net  is safe, 
  the option to complete implies proper completion. 
\end{lemma}
\begin{proof}
  Let \(\M[']\) be a reachable marking that covers the final place, %
  i.e., \({}o \in \M[']\). %
  The existence of %
  an option to complete implies that %
  there is a firing sequence from the marking~\(\M[']\) to the final marking, %
  i.e., a firing sequence of the form.
  \[
    \M['] \fire{t_1} \M[1] \fire{t_2} \cdots \M[n]\fire{t_{n}}\thefin
  \]
  If this firing sequence is empty, %
  i.e., if~\(n=0\), %
  we obtain the desired, %
  because then we have already started from the final marking,
  i.e., \(\M[']=\thefin\). %
  It remains to show that, %
  for every reachable marking~\(\M[']\) that covers~\({}o\), %
  there cannot be any (non-empty) firing sequence 
  \[
    \M['] \fire{t_1} \M[1] \fire{t_2} \cdots \M[n]\fire{t_{n}}\thefin
  \]
  of length~\(n\geq 1\).%
  We shall proceed by induction on~\(n\geq 1\).  
  \begin{description}
  \item[\(n=1\)]
    Let \(\M[']\) be a reachable marking such that \(o \in \M[']\) %
    and assume that \(\M['] \fire{t_1} \thefin\) %
    (to derive a contradiction).
    Because the net is an elementary net system, %
    the transition~\(t_1\) consumes a token from the marking~\(\M[']\) %
    and produces a token in the final place. %
    However %
    this would mean that the net is not safe %
    as the marking \(\M[']\)~covers the final place already %
    and \(t_1\) cannot consume a token from the final place,
    by the definition of elementary workflow net. %
    This however leads to a contradiction to the assumption of safety.
    Hence, %
    there cannot be a firing sequence of length one from %
    the marking~\(\M[']\) to the final marking. %
  \item [\(n\leadsto n+1\)]
    The induction hypothesis is that,
    for every  reachable marking~\(\M[']\) with \(o \in \M[']\), %
    there cannot be any firing sequence 
    \[
      \M['] \fire{t_1} \M[1] \fire{t_2} \cdots \M[n]\fire{t_{n}}\thefin
    \]
    of length~\(n\geq1\) %
		to the final marking. %
    Let \(\M['']\)~be a reachable marking that such that~\(o \in \M['']\) %
    and assume that there exists a firing sequence
    \[
     \M[''] \fire{t'_1} \M[1]' \fire{t'_2} \cdots \M[n+1]'\fire{t'_{n+1}}\thefin
    \]
    (to derive a contradiction). %
    Note that %
    \(\M[1]'\) still covers the final place (as \({}o \notin \pre[t'_1]\)) %
    and is reachable. %
    This however leads to a contradiction to the induction hypothesis, %
    which states that %
    there cannot by any firing sequence of the form
    \[
      \M[1]' \fire{t'_2} \cdots \M[n+1]'\fire{t'_{n+1}}\thefin
    \]
    with \(\M[1]'\) reachable and covering the final place. %
    Thus, %
    we have derived a contradiction. %
  \end{description}
  Finally, %
  by the principle of complete induction, %
  there cannot be any firing sequence 
  \[
    \M['] \fire{t_1} \M[1] \fire{t_2} \cdots \M[n]\fire{t_{n}}\thefin
  \]
  of length~\(n\) %
  for \(n>0\). %
  Hence \(n=0\) is the only relevant case, %
  which we have already covered, %
  and thus the proof is complete. %
  \qed
\end{proof}


We now provide the proof of Theorem~\ref{thm:da-thm}.
To avoid clutter, %
we denote the restart game by \(\game\). %

\medskip
\par\noindent \textbf{Theorem 1 (Characterization of soundness).}
\emph{  Let \(\Net\) be %
  a elementary  workflow net that is safe and 
  extended free choice;
  let \((\Net,\rl \colon T \to \{\Sigma\},\underline 1)\)
  be the net with transition payoffs and roles where %
  \(\Sigma\) is a unique role,
  \(\rl\colon T \to \{\Sigma\}\) is the unique total function, %
  and \(\underline 1\) is the constant utility-1 function. %
  The soundness property holds in \(\Net\) if, and only if, %
  incentives in \((\Net,\rl \colon T \to \{\Sigma\},\underline 1)\)  %
  are aligned w.r.t.\ full liveness. %
	
}
\begin{proof}
  Let \(\Net\) be an %
  elementary workflow net that is safe and extended free-choice
  with initial marking~\(\theini\). %
  This implies %
  that the set of transitions is non-empty, %
  because the final place has to be an element of some post-set, %
  by definition. %

  First, assume that %
  the net \((\Net,u,\rho)\) with initial marking~\(\theini{}\) is %
  incentive aligned w.r.t.\ full liveness. %
  By definition of the latter, %
  for some specific choice of \(\varepsilon>0\) %
  (e.g., \(\varepsilon=0.5\)), %
  there exists an eventually positive,
  correlated \(\varepsilon\)-equilibrium~\(\sigma\) %
  of the {restart game} \(\game\) %
  such that,
  for every transition~\(t \in \tr\), %
  every reachable marking~\(\M[']\), and  %
  for every natural number~\(\nnull\in\mathbb{N}\), %
  there exists a history %
  \begin{equation}
    \label{eq:da-history}
    h=\langle \M['],\m_1,\act_1, \dotsc,\state[n-1],\m_{n-1},\act_{n-1}, \state[n],\m_n \rangle \in \hist{n}
  \end{equation}
  at stage~\(n> \nnull\) with %
  \(\act_{n-1}=\{t\}\)
  (using that there are no parallel steps allowed in the restart game, and there is only one role) 
  and \(\PP{\M[']}{\sigma}(h) >0\). %

  We first address the non-existence of \textbf{dead transitions}. %
  Let \(t\in\tr\)~be a transition. %
  We have to show that~\(t\) it is not dead. i.e., %
	will eventually be fired. %
  Now,
  \(\M[0]=\theini{}\) is a reachable marking %
  and \(1\)~is a natural number. %
  Hence, %
  by the above property of~\(\sigma\), %
  there exists a history
  \begin{equation}
    \label{eq:da-history-no-dead}
    h=\langle \M[0],\m_1,\act_1, \dotsc,\state[n-1],\m_{n-1},\act_{n-1}, \state[n],\m_n \rangle \in \hist{n}
  \end{equation}
  at stage \(n>0\) with
  with \(\act_{n-1}=\{t\}\) and \(\PP{\M[0]}{\sigma}(h) >0\). %
  Let~\(k<n\) be the largest index such that \(\state[k] = \M[0]\). %
  Thus, %
  in the Petri net~\(\Net\), %
  there is a firing sequence 
  \[
    \theini{}=\M[0]=\state[k] \fire{\act_k}  \cdots \state[n-1]\fire{\act_{n-1}}\state[n] 
  \]
  such that~\(\act_{n-1}=\{t\}\); %
  hence, %
  \(\state[n-1]\fire{t}\) and \(\state[n-1]\) is reachable. %
  As the transition~\(t\) was arbitrary, %
  incentive alignment w.r.t. full liveness implies that
	there are no dead transitions. %

  Concerning the \textbf{option to complete}, %
  let \(\M[']\)~be an arbitrary reachable marking. %
  For the option to complete,
  we want to show that there exists
  a firing sequence
  \[\M[']\fire{t_1}\cdots\M[n]\fire{t_n}\thefin\]
  from the reachable marking~\(\M[']\) to the final marking~\(\thefin\). %
  First, %
  note that we can assume w.l.o.g.\ that \(\M[']\neq\thefin\) %
  (as otherwise the empty firing sequence suffices). %
  Hence, %
  we only need to consider the case in which \(\M[']\neq\thefin\). %
  We proceed by choosing a transition~\(t_0\in\tr\) that is enabled at the initial marking, %
  i.e., \(\theini{}\fire{t_0}\); %
  such a transition must exist %
  since the net~\(\Net\) is an elementary net system.
  Also,  \(2\)~is a natural number.\thnote{2 is magic \tt\^{}\_\^{}.} %
  Hence, %
  by the above property of~\(\sigma\), %
  there exists a history
  \begin{equation}
    \label{eq:da-history-again}
    h=\langle \M['],\m_1,\act_1, \dotsc,\state[n-1],\m_{n-1},\act_{n-1}, \state[n],\m_n \rangle \in \hist{n}
  \end{equation}
  at stage~\(n> 2\) with %
  \(\act_{n-1}=\{t_0\}\)
  and \(\PP{\M[']}{\sigma}(h) >0\).%
  \iwnote{aber t0 ist doch die aller erste Transition, as per above?\\
    \ul{TH} %
    restart game for the win \\
    \ul{IW} Jaja, schon, aber es kann doch m' = [i] sein? und dann war noch nix restart, und das Argument greift nicht\\
    \ul{TH} \emph{aiiiii} -- index \(\pm1\)-Fehler !!! Da hab ich mich bei den indices um 1 vertan.
    Nach der Änderung darf auch \(\M[']=\M[0]=\theini\) gelten. 
  }
  Hence, %
  \(\state[n-1] = \M[0] = \theini\), %
  because the net is an elementary workflow net and \(\pre[t_0]=\theini=\M[0]\). %
  Thus, %
  there exists a smallest index~\(\ell>1\) such that
  \(\state[\ell]=\M[0] =\theini\). 
  Hence, %
  there is a non-empty firing sequence %
  \[
    \M['] \fire{\act_1}  \cdots \state[\ell-1]\fire{\act_{\ell-1}}\thefin
  \]
  in the Petri net~\(\Net\), %
  by definition of the restart game and because \(\thefin\) is the only final marking.  
  Thus, %
  we have shown that options to complete exist. %
  Proper completion follows by Lemma~\ref{lem:proper-completion}. 


  \textbf{Conversely}, %
  assume that the net~\(\Net\) satisfies the soundness property. %
  We choose a correlation device~\(\mathcal{D}\) that gives the same signal %
  (no matter the history). %
  Thus, %
  \(\sig{\Sigma}{n}=\{\top\}\) where \(\top\) is an arbitrary unique signal %
  and the probability distributions \(\sigd[n]\vect \m {n-1}\)
  concentrate all probability mass on the unique signal~\(\m=\langle\top\rangle\). %
  
  First, note that 
  if player \(\Sigma\) always chooses some transition as action, %
  or, equivalently, never chooses to idle, %
  we obtain a strategy profile for the correlation device, %
  which is an eventually positive \(\varepsilon\)-equilibrium strategy profile.%
  \footnote{
    These strategies are the best ones existing %
    or \emph{dominant}, %
    in game theoretic vernacular. 
  }
  \ Here, %
  we use the fact that the restart game allows %
  the player to execute at most one transition at a time. %
  Now, %
  relying on the option to complete, 
  we shall construct a strategy~\(\sigma\) that %
  only considers the current state; %
  such strategies are called \emph{stationary}. %
  Thus, %
  let~\(\sigma\) be \emph{the} strategy of player~\(\Sigma\) that chooses, uniformly at random, %
  one of the transitions that are enabled at the current state. %
  More formally, %
  if we put \(L_{\M} = \{ t \in \tr \mid \M\fire{t}\}\),
  for any reachable marking~\(\M\), %
  then~\(\sigma\) is characterized by the equation \[\sigma^\Sigma_h \{t\} =
    \begin{cases}
      \frac{1}{L_{\state[h]}} & \text{if } \state[h] \fire{t} \\
      0 & \text{otherwise},
    \end{cases}
  \]
  for every observation~\(h\). %
  
  Now, %
  we claim that the strategy~\(\sigma\) yields a ``uniform'' witness for %
  incentive alignment w.r.t.\ full liveness, %
  i.e., that for every \(\varepsilon>0\), %
  the strategy profile \(\langle\sigma\rangle\) is an %
  eventually positive correlated \(\varepsilon\)-equilibrium %
  such that,  %
  for all transitions \(t \in \tr\), %
  every reachable marking~\(\M[']\), and %
  each natural number~\(\nnull \in \mathbb{N}\), %
  there exists a history %
  \[h=\langle \M['],\m_1,\act_1, \dotsc,\state[n-1],\m_{n-1},\act_{n-1}, \state[n],\m_n \rangle \in \hist{n}
  \]
  at stage~\(n> \nnull\)  with %
  \(t \in \act_{n-1}\) and \(\PP{\M[']}{\sigma}(h) >0\). %
  
  Thus, %
  let \(\varepsilon >0\)~be a positive real, %
  let \(t\in\tr\)~be a transition, %
  let \(\M[']\)~be a reachable marking, and %
  let \(\nnull\in\mathbb{N}\)~be a natural number. %
  We first show that there exist 
  \begin{enumerate}
  \item %
    a history~\(h'\) that starts from \(\M[']\) and leads to~\(\theini{}\), %
  \item %
    a history~\(h''\) that starts from \(\theini{}\), also leads to~\(\theini{}\), and is non-empty, and
  \item 
    a  history~\(h'''\) that starts from \(\theini{}\) and leads to a marking that enables the transition~\(t\) %
  \end{enumerate}
  and all these histories are ``possible'', %
  i.e., \(\PP{\M[']}{\sigma}(h') >0\),
  \(\PP{\theini}{\sigma}(h'') >0\), %
  and \(\PP{\theini}{\sigma}(h''') >0\). %
  \paragraph{Back to the initial marking} 
  By construction of~\(\sigma\), %
  using the option to complete and Lemma~\ref{lem:proper-completion}, %
  there is a history~\(h'\) of the form %
  \[
    h'=\langle \M['],\m_1',\act_1', \dotsc,\state[k-1]',\m_{k-1}',\act_{k-1}', \state[k]',\m_k' \rangle \in \hist{k}
  \]
  with \(\PP{\M[']}{\sigma}(h') >0\) %
  such that
  \(\state[k]' = \theini{}\), %
  \(\state[l]'\neq \theini{}\) for all~\(l<k\), %
  and \(k>0\). %
  \paragraph{Looping on the initial marking}
  There is an initial transition~\(t_0\), %
  i.e., one such that \(\pre[t_0]=\theini{}\); %
  let \(\M['']\) be the marking reached after firing \(t_0\), %
  i.e., \(\theini{}\fire{t_0}\M['']\). %
  Now, %
  by the same argument as for~\(h'\), %
  we obtain
  a history~\(h''\) of the form %
  \[h''=\langle \M[''],\m_1'',\act_1'', \dotsc,\state[r-1]'',\m_{r-1}'',\act_{r-1}'', \state[r]'',\m_r'' \rangle \in \hist{r}
  \]
  with \(\PP{\M['']}{\sigma}(h'') >0\) %
  such that 
  \(\state[r]'' = \theini{}\), %
  \(\state[j]''\neq \theini{}\) for all~\(j<r\), %
  and \(r>0\). %
  \paragraph{Reaching the transition}
  Finally, %
  by construction of~\(\sigma\), %
  and using the absence of dead transitions, %
  there is a history~\(h'''\) of the form %
  \[h'''=\langle \M['''],\m_1''',\act_1''', \dotsc,\state[n-1]''',\m_{n'-1}''',\act_{n'-1}''', \state[n']''',\m_{n'}''' \rangle \in \hist{n'}
  \]
  with \(\PP{\M[''']}{\sigma}(h''') >0\) %
  such that \(\{t\} = \act_{n'-1}'''\). %

  \paragraph{Putting histories together}
  Finally, %
  as the strategy~\(\sigma\) is stationary, %
  i.e., it does not distinguish between histories with the same current state, %
  we can concatenate\thnote{%
    yes,  technically this concatenation is not defined.
    Can we leave this as the test whether a reviewer read the proof to the end? \\
		\ul{IW} I think that is basic / obvious enough.
  }~\(h'\), %
  enough copies of~\(h''\), %
  and the third history \(h'''\) %
  to obtain a history~\(h\) %
  as desired. %

  As \(\sigma\) always chooses some transition to execute next, %
  it is a dominant strategy and thus in particular an \(\varepsilon\)-equilibrium; %
  moreover, %
  it is eventually positive since it reaches average utility of~\(1\). %
  \qed
\end{proof}

}%

\color{white}
\begin{iflong}
  \smash{\makebox[0pt][l]{if long}}
\end{iflong}
\begin{ifshort}
  \smash{\makebox[0pt][l]{if short}}
\end{ifshort}
\end{document}